\documentclass[a4paper,twoside,11pt]{article}

\usepackage[utf8]{inputenc}

\usepackage{amscd}
\usepackage{amsmath,amssymb,amsthm,mathrsfs, mathtools}
\usepackage[margin=2.5cm]{geometry}

\usepackage{tikz} 
\usetikzlibrary{matrix,arrows,calc,through}
\tikzset{bvert/.style={draw,circle,fill=black,minimum size=5pt,inner sep=0pt}  } 
\tikzset{wvert/.style={draw,circle,fill=white,minimum size=5pt,inner sep=0pt}  } 
\tikzset{ivert/.style={rectangle,minimum size=3pt,inner sep=1pt}  } 

\usepackage{parskip}
\usepackage{enumerate}
\usepackage{enumitem}
\setlist[enumerate, 1]{label=(\roman*)}
\setlist[itemize]{leftmargin=1.5em}
\setlist[description]{leftmargin=1em}

\usepackage{mathtools} 

\usepackage[hidelinks, backref=page]{hyperref}  
\usepackage[format=plain,labelsep=period,justification=justified,font=small,labelfont=bf, margin=1em]{caption}

\newcommand{\lat}{{\mathcal{L}}}
\newcommand{\hlat}{{\hat{{\mathcal{L}}}}}
\newcommand{\oct}{{\mathcal{O}}}
\newcommand{\tet}{{\mathcal{T}}}
\newcommand{\tetb}{{\mathcal{T}^\bullet}}
\newcommand{\tetw}{{\mathcal{T}^\circ}}

\newcommand{\ub}{p^\bullet}
\newcommand{\uw}{p^\circ}
\newcommand{\Xb}{{X^\bullet}}
\newcommand{\Xw}{{X^\circ}}
\newcommand{\Wb}{{W^\bullet}}
\newcommand{\Ww}{{W^\circ}}

\newcommand{\medminus}{\scalebox{0.6}[0.7]{\(-\)}}
\newcommand{\sm}{\!\medminus}
\newcommand{\shi}[1]{{\sigma_{#1}}}
\newcommand{\shim}[1]{{\sigma_{\sm#1}}}
\newcommand{\hC}{{\hat {\mathbb C}}}
\DeclareMathOperator{\mr}{\mathrm{mr}}
\DeclareMathOperator{\cro}{cr}

\newcommand{\N}{\mathbb{N}}
\newcommand{\Z}{\mathbb{Z}}
\newcommand{\R}{\mathbb{R}}
\newcommand{\C}{\mathbb{C}}
\newcommand{\CP}{\mathbb{C}\mathrm{P}}

\usepackage{amsthm}

\theoremstyle{plain}

\theoremstyle{definition}
\newtheorem{theorem}{Theorem}[section]
\newtheorem{corollary}[theorem]{Corollary}
\newtheorem{lemma}[theorem]{Lemma}

\newtheorem{definition}[theorem]{Definition}
\newtheorem{question}[theorem]{Question}

\theoremstyle{remark}
\newtheorem{remark}[theorem]{Remark}

\title{\vspace{-2.3cm}Möbius invariant Y-systems \\ (cluster structures) for Miquel dynamics}

\author{Niklas C. Affolter \bigskip\\
	Institut f\"ur Mathematik, TU Berlin, \\
	Straße des 17.\@ Juni 136, 10623 Berlin, Germany \bigskip\\	
	Institut für Diskrete Mathematik und Geometrie, TU Wien,  \\
	Wiedener Hauptstraße 8-10/104, 1040 Wien, Austria	
}

\date{October 23, 2024}

\begin{document}

\maketitle

\noindent
\textbf{Abstract.} Miquel dynamics is a discrete time dynamics for circle patterns, which relies on Miquel's six circle theorem. Previous work shows that the evolution of the circle centers satisfy the dSKP equation on the octahedral lattice $A_3$. As a consequence, Miquel dynamics is a discrete integrable system. Moreover, Miquel dynamics give rise to a real-valued cluster structure. The evolution of the cluster variables under Miquel dynamics is also called a Y-system in the discrete integrable systems community. If the Y-system is real positive-valued then the circle pattern is accompanied by an invariant dimer model, an exactly solvable model studied in statistical physics. However, while circle patterns are Möbius invariant, the circle centers and the Y-system are not Möbius invariant, which violates the so called transformation group principle. In this article we show that half the intersection points satisfy the dSKP equation as well, and we introduce two new real-valued Y-systems for Miquel dynamics that involve only the intersection points. Therefore, the new Y-systems are Möbius invariant, and thus satisfy the transformation group principle. We also show that the circle centers and intersection points combined satisfy the dSKP equation on the 4-dimensional octahedral lattice $A_4$. In addition, we present two more complex-valued Y-systems for Miquel dynamics, which are real-valued in and only in the case of integrable circle patterns. We also investigate the special cases of harmonic embeddings and s-embeddings, which relate to the spanning tree and Ising model respectively.

\tableofcontents

\newpage

\section{Introduction}

A \emph{circle pattern} is a map from $\Z^2$ to the complex plane, such that the image of each unit-square is inscribed in a circle. \emph{Miquel dynamics} is a discrete time dynamics on circle patterns. In other words, given an initial circle pattern, Miquel dynamics is a rule to produce a bi-infinite sequence of circle patterns from the initial circle pattern. In alternating fashion, Miquel dynamics replaces every second circle with the circle through the other four intersection points of the adjacent circles see Figure~\ref{fig:miquelcps}. The fact that such a circle always exists is a classic theorem in circle geometry called \emph{Miquel's theorem} \cite{amiquel}.
Miquel dynamics was invented by Richard Kenyon and formally introduced  by Ramassamy \cite{ramassamymiquel}. It was conjectured that Miquel dynamics is in some sense integrable. In the biperiodic $2\times 2$ case, a first integrability result was obtained by Glutsyuk and Ramassamy \cite{grmiquel}.

The first general integrability results were obtained by the author \cite{amiquel} and independently by Kenyon, Lam, Ramassamy and Russkikh \cite{klrr}. The first observation was that the circle centers constitute a \emph{dSKP lattice}, that is a solution of the \emph{dSKP equation}, short for \emph{discrete Schwarzian Kadomtsev–Petviashvili equation}. This equation is one of the five discrete integrable equations -- in the sense of multi-dimensional consistency -- on the octahedral lattice \cite{abs}. The second observation was that it is possible to associate certain $Y$-variables to a circle pattern, such that the $Y$-variables change under Miquel dynamics according to the rules of \emph{mutation} in a \emph{cluster algebra} \cite{fzclusteralgebra}. This is why the $Y$-variables are also called \emph{cluster variables}. Let us add that in the discrete integrable systems community, the equivalent observation is that the evolution of the $Y$-variables under Miquel dynamics constitutes a \emph{Y-system}. Moreover, in a general dSKP lattice the $Y$-variables are complex-valued. The third observation was that the $Y$-variables are \emph{real-valued} if and only if the dSKP lattice corresponds to the centers of a circle pattern. Moreover, with some additional embedding assumptions on the circle pattern, the $Y$-variables are actually \emph{positive real}. As a result, the $Y$-variables may be used to define a \emph{dimer model}, an exactly solvable model defined studied in probability theory \cite{kenyondimerintro}. The partition functions of this model turn out to be invariants of Miquel dynamics. The dimer models associated to circle patterns in this way have been the subject of further study \cite{clrlorentz, clrtembeddings, bnrtemb, craztec}. Let us add that Kenyon et al.\ \cite{klrr} have also provided existence and uniqueness results for circle patterns with given $Y$-variables.

Circle patterns are considered to be a discretization of conformal maps, going back to ideas of Koebe \cite{koebecp} and Thurston \cite{thurstoncp}. For example, one may consider two circle patterns to be discrete conformally equivalent if they have the same intersection angles, see \cite{bsvariationalcp}. Or one may consider circle patterns to be discrete conformally equivalent if they have the same length cross-ratios, see  \cite{bpsdiscretehyperbolic}. As is well known, conformal maps are invariant under Möbius transformations. The so called \emph{transformation group principle} for structure-preserving discretizations states that discretizations should exhibit the same transformation invariance as the smooth object \cite{bsorganizing}. And indeed, the two aforementioned discretizations of conformal maps are Möbius invariant. From that perspective, the current description of Miquel dynamics is insufficient: the dSKP lattice associated to Miquel dynamics is not Möbius invariant. And the $Y$-variables associated to a dSKP lattice as in \cite{amiquel,klrr} are also not Möbius invariant. This also implies that the associated dimer model is not Möbius invariant.

The main purpose of this paper is to present a dSKP lattice and a real-valued Y-system for Miquel dynamics which are both Möbius invariant and thus adhere to the transformation group principle. 
In the first step, we show that either half of the intersection points of the circle patterns under Miquel dynamics are also a dSKP lattice. Since the intersection points of a circle pattern are Möbius invariant, this dSKP lattice is Möbius invariant. In the second step, we introduce the \emph{$X$-variables}, which are given by certain cross-ratios of intersection points. We show that the $X$-variables constitute a Y-system. Since cross-ratios are Möbius invariant, so are the $X$-variables. We show that this Y-system is a specialization of the one introduced by Adler, Bobenko and Suris \cite{abs}. In fact, we show that a dSKP lattice corresponds to half the intersection points of a circle pattern if and only if the $X$-variables are real. Additionally, in upcoming work with Mühlhoff \cite{ammiquel} we are able to show that circle patterns which satisfy the same embedding assumptions as in \cite{klrr} have $X$-variables that are \emph{positive real}. 

There are other geometric dynamics besides Miquel dynamics for which analogous cluster algebra results were obtained, for example: the \emph{pentagram map} \cite{glickpentagram}, \emph{polygon recutting} \cite{izosimovrecutting} and \emph{cross-ratio dynamics} \cite{agrcr}. Note that in each case there is a natural transformation group that commutes with the dynamics: $\mathrm{PGL}(3, \R)$, $\mathrm {Aff}(\C)$ and $\mathrm{PGL}(2, \C)$, respectively. In each case, the cluster variables are invariant under these transformations. In that sense, the new description of Miquel dynamics via $X$-variables is more in line with the literature on other examples than the previous description of Miquel dynamics via $Y$-variables.

Note that to a circle pattern under Miquel dynamics, we are now able to associate two different dSKP lattices via each half of the intersection points, and one dSKP lattice via the circle centers. Moreover, we show that to each dSKP lattice, there are actually two independent sets of $X$-variables, and therefore two different Y-systems. Additionally, there is the Y-system as introduced in \cite{amiquel, klrr}. Hence, there are $9 = 3\times 3$ Y-systems associated to Miquel dynamics in total, although only three of those are real-valued. We think it is an interesting aspect of our results, that a dynamical system may come with several Y-systems instead of just one, with different transformation symmetries and some of them real-valued, some complex-valued. In particular, this shows that there is not necessarily a canonical Y-system. Of course, it is natural to ask how all these lattices and variables are related. On the level of dSKP lattices, we are able to show that the three dSKP lattices may actually be considered to be \emph{one} dSKP lattice but on a 4-dimensional octahedral lattice instead of a 3-dimensional one. Thus we give a unifying description of the dSKP lattice in \cite{amiquel, klrr} and the new ones we present here.

We also think that the unified view on the dSKP lattices may prepare the way to solve another interesting question: do the various Y-systems also satisfy some functional relations? Let us point out that in \cite{athesis} there is a general method to associate a whole sequence of Y-systems to many geometric systems. However, it is not immediately clear if this method relates to the different Y-systems presented here. The question of relations between the Y-systems is also of particular interest because each Y-system -- viewed as a cluster structure -- comes with a canonical Poisson structure \cite{gsvpoissonpaper}. Moreover, in the biperiodic case there is also a full set of Hamiltonians for this Poisson structure, via the \emph{dimer integrable system} \cite{gkdimers}. Hence, one may also ask how these Poisson and Hamiltonian structures are related. It might turn out that they are equivalent, or that they give rise to a bi-Hamiltonian structure, or something else altogether.

Finally, there are some cases of circle patterns which are of special interest with regard to Miquel dynamics and also independently in the literature. The first special case are \emph{circle packings} \cite{stephensoncp}, the center nets of which are called \emph{s-embeddings} \cite{chelkaksembeddings, chelkaksgraphs}. The second special case are circle patterns related to \emph{harmonic embeddings} or \emph{Tutte embeddings} \cite{biskupharmonic, tutteembedding}. It was Kenyon et al.\ \cite{klrr} that showed how in these two special cases the $Y$-variables satisfy certain local constraints, which imply that the dimer model may be specialized to the \emph{Ising model} and the \emph{spanning tree model} respectively, see \cite{klrr} and \cite{clrtembeddings} for more background. Thus the study of these two special cases is primarily motivated from the statistical mechanics point of view. Of course, circle packings are also of general interest in discrete analytic function theory, see \cite{stephensoncp}. Note that these two special cases are \emph{not} invariant under Miquel dynamics. From the integrable systems perspective, s-embeddings are characterized by the fact that the evolution of Miquel dynamics has a global (time-)symmetry, which is also reflected in the $Y$-variables. As a consequence of the symmetry, we show that for s-embeddings the two different $X$-variables coincide. In the case of harmonic embeddings it is not clear if there is some global symmetry. Instead, we show that in this case one set of $X$-variables coincides with the $Y$-variables. The third special case we consider are \emph{integrable circle patterns} \cite{bmsanalytic}, which were introduced as an integrable discretization of discrete analytic functions. It turns out that integrable circle patterns are preserved under Miquel dynamics. Moreover, we show that in this case two more Y-systems are real-valued compared to the general case. We also show that the two -- now real-valued -- Y-systems actually coincide. The proof relies on an apparently novel identity for the two cross-ratios in a Miquel cube (Theorem~\ref{th:wconjugate}). Let us add there is another special case of circle patterns with factorizing cross-ratios \cite{bpdisosurfaces}, that comes from the theory of discrete isothermic surfaces. It is currently unclear if this special case has particular properties with respect to Miquel dynamics.

\subsection*{Plan of the paper} 

In Section~\ref{sec:main} we present the main results of the paper. We focus on Miquel dynamics as yielding a sequence of circle patterns, a viewpoint close to the previous descriptions of Miquel dynamics in \cite{amiquel,klrr} and to the descriptions of similar systems in the cluster algebra community.
In Section~\ref{sec:octahedralcombinatorics}, we introduce the combinatorics of $A_3$ as viewed on $\Z^3_+$, which we denote $\lat$, and explain the occurrence of tetrahedra $\tet$ and octahedra $\oct$ in $\lat$. In Section~\ref{sec:miquelmaps}, we put the evolution of Miquel dynamics on the lattice $\lat$, we call the resulting objects \emph{Miquel maps}. In particular, we associate circles and centers with $\lat$, and intersection points with the tetrahedra $\tet$. This is the lattice viewpoint which is more common in the discrete integrable systems community, and which is useful to provide precise statements and proofs of our results. In Section~\ref{sec:dskp}, we explain and prove the dSKP results for Miquel maps on $\lat$. In Section~\ref{sec:afour}, we show how to identify $\tet \cup \lat$ with a subset of $A_4$. Then we prove that with this identification, the combined map of centers and intersection points solves the dSKP equation on $A_4$. In Section~\ref{sec:toda}, we explain the Y-system results for $Y$- and $X$-variables on $\lat$. In Section~\ref{sec:icp}, we introduce what we call \emph{$W$-variables} for Miquel maps, which are also a Y-system, but a complex-valued one. We show that the special case of real $W$-variables is the case of so called \emph{integrable circle patterns}. Finally, in Section~\ref{sec:harmonic} we show that $X$- and $Y$-variables coincide in the special case of $h$-embeddings and in Section~\ref{sec:packings} we discuss circle packings and $s$-embeddings and the resulting symmetries in the $X$-, $Y$- and $W$-variables.

\subsection*{Acknowledgments}
The author is supported by the Deutsche Forschungsgemeinschaft (DFG) Collaborative Research Center TRR 109 ``Discretization in Geometry and Dynamics''. I would like to thank Sanjay Ramassamy, Wolfgang Schief and Boris Springborn for discussions on the topic. I would also like to thank the anonymous referees for their thorough reviews which helped to improve the presentation of the paper significantly.

\section{Main results} \label{sec:main}

\begin{figure}
	\centering
	\includegraphics[scale=0.48]{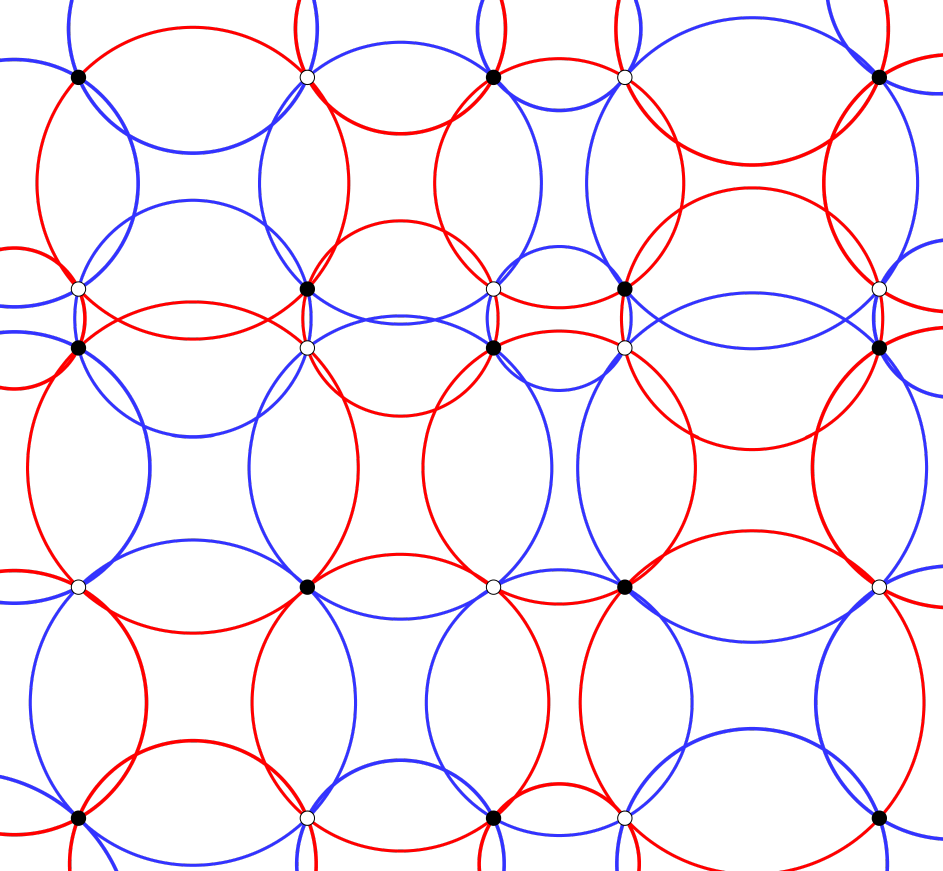}
	\hspace{4mm}
	\includegraphics[scale=0.48]{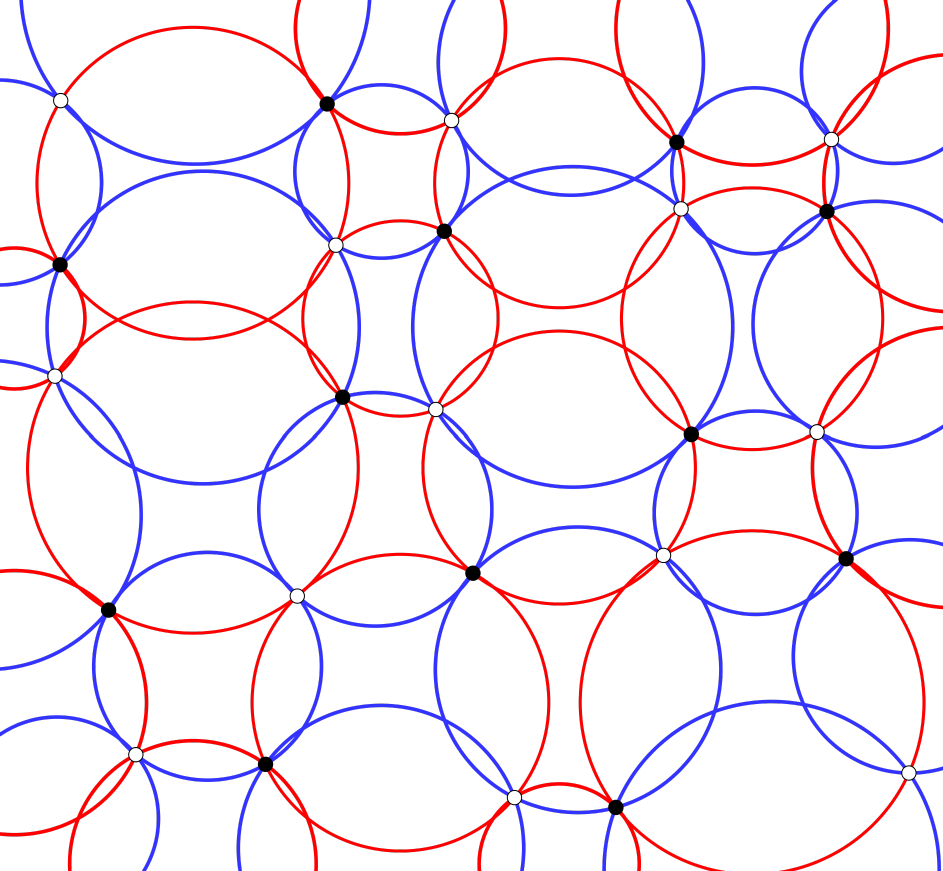}
	\caption{Miquel dynamics transforms the circle pattern $(c_0,p_0)$ on the left into the circle pattern $(c_1,p_1)$ on the right, by replacing all the blue circles.}
	\label{fig:miquelcps}
\end{figure}

Let $F(\Z^2)$ denote the faces of $\Z^2$, and let $\hC = \C \cup \{\infty\}$ denote the \emph{complex projective line}. As a set $\hC$ is also called the \emph{extended complex numbers}. Moreover, via stereographic projection $\hC$ is also in bijection with the \emph{Riemann sphere}, which is also sometimes denoted by $\hC$. We use the standard rules of arithmetic on $\hC$. We call a subset of $\hC$ a \emph{circle} if it is a circle in $\C$, or if it consists of the union of $\{\infty\}$ and a line in $\C$.

\begin{definition} \label{def:cp}
	A \emph{circle pattern} $(c,p)$ is a pair of maps
	\begin{align}
		c&: \Z^2 \rightarrow \mbox{Circles}(\hC),\\
		p&: F(\Z^2) \rightarrow \hC, 
	\end{align}
	such that $p(f) \in c(v)$ whenever $v \in f$. 
\end{definition}

It may seem counterintuitive to associate circles to vertices of $\Z^2$ and points to faces. However, for the dynamics we study it will prove useful to view $\Z^2$ as a subset of a higher-dimensional lattice, and for that purpose it is more practical to define circle patterns by associating circles to vertices of $\Z^2$ and points to faces.  

\emph{Miquel dynamics} \cite{ramassamymiquel} is a discrete time dynamics on circle patterns, and the evolution of Miquel dynamics produces a sequence of circle patterns $(c_k, p_k)_{k\in\Z}$, see also Figure~\ref{fig:miquelcps}. Let us explain how Miquel dynamics are defined, beginning with a circle pattern $(c_0,p_0)$. The vertices of $\Z^2$ may be partitioned into \emph{even vertices} and \emph{odd vertices}, where even and odd refer to the parity of the coordinate sum. Consider an even vertex $v\in V$ and its four adjacent vertices $v_1,v_2,v_3,v_4$ in counterclockwise order, see Figure~\ref{fig:miquellabels}. Let $f_{12}$ be the face incident to $v,v_1,v_2$, and $f_{23}$ the face incident to $v,v_2,v_3$, and $f_{34}$ and $f_{41}$ analogously, see also Figure~\ref{fig:miquellabels}.
The circles $c_0(v_1)$ and $c_0(v_2)$ have one intersection point $p_0(f_{12})$ on $c_0(v)$, and generically they have a second intersection point $p_1(f_{12})$. The exception is if $c_0(v_1),c_0(v_2)$ are tangent, in which case we set $p_1(f_{12}) = p_0(f_{12})$. Throughout we assume it is understood that in the case of tangent circles the second or the other intersection point is the same as the first intersection point.
\emph{Miquel's theorem} \cite{miquel} states that the four points $p_1(f_{12}), p_1(f_{23}), p_1(f_{34}), p_1(f_{41})$ are on a circle, which we denote by $c_1(v)$. By setting $c_1(v) = c_0(v)$ for all odd vertices, we obtain a new circle pattern $(c_1,p_1)$. We say the two circle patterns are related by Miquel dynamics. Note that this operation is invertible by applying the same procedure to the even circles again. However, applying the substitution to all odd circles yields a new circle pattern $(c_2, p_2)$. Therefore, by alternating even and odd substitutions we obtain a bi-infinite sequence of circle patterns $(c_k, p_k)_{k\in\Z}$, which is the forwards and backwards evolution of Miquel dynamics.

Let us denote by $t_k: \Z^2 \rightarrow \C$ the circle centers of $c_k$. We call a map $t_k$ that arises as the circle centers of a circle pattern a conical net, see for example \cite{muellerconical}. We do not need this characterization, but conical nets may be characterized by an angle condition, specifically that the two sums of opposite angles are equal. These nets were originally studied in $\R^3$ \cite{lpwywconical}, and the term ``conical'' appears because the angle condition implies that certain planes touch a cone.

Conical nets are also called \emph{Coloumb gauge} in \cite{klrr}, since they interpret the edge-vectors $t(v)-t(v')$ of a conical net as the edge-weights of a dimer model in a special complex gauge: the Coulomb gauge.
Conical nets are also called \emph{$t$-realization} in \cite{clrtembeddings}, since they are related to \emph{T-graphs} \cite{kenyonsheffield}, and also because the authors introduce a concept of \emph{$t$-holomorphicity} that generalizes \emph{$s$-holomorphicity} \cite{smirnovtowardsconformal, smirnovrandomcluster, csisoradial}. The concept of $s$-holomorphicity was used to obtain the celebrated conformal invariance of the 2D Ising model \cite{csising}. A special case of $t$-realizations are \emph{$t$-embeddings}, which are conical nets with an additional embedding property. More specifically, a $t$-embedding is a conical net such that the images of the faces are convex and non-overlapping, and the images of adjacent vertices are distinct.
Note that the notions of Coulomb gauge and $t$-embeddings were developed independently at the same time, which explains the two different terms. Moreover, it appears that the respective authors were not aware of the older and more established term of conical nets introduced in \cite{lpwywconical} and in the planar case in \cite{muellerconical}, see also \cite[Chapter~3.3]{ddgbook} for additional references.

\begin{figure}
	\centering
	\includegraphics[scale=0.4]{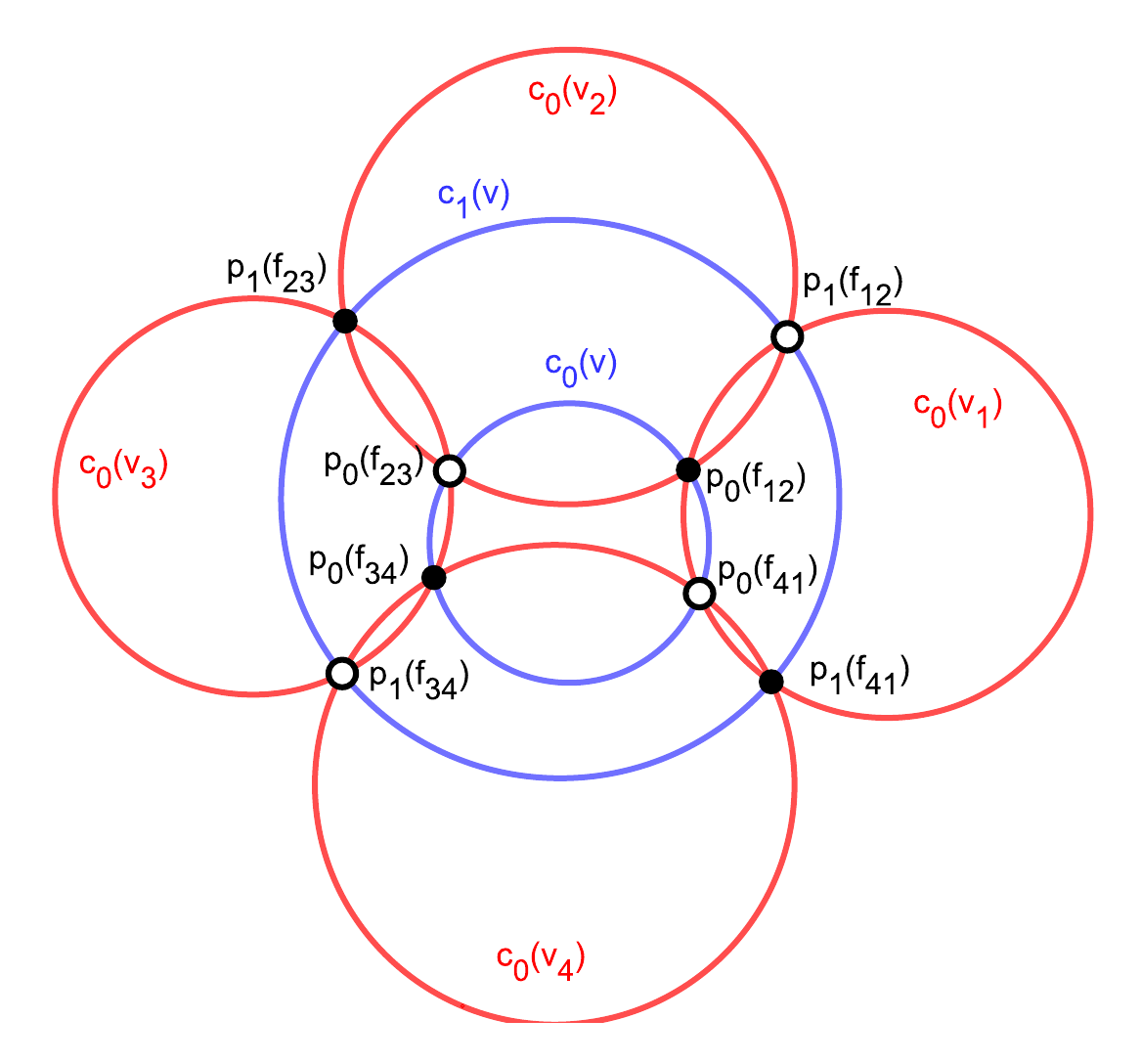}
	\hspace{-2mm}
	\includegraphics[scale=0.4]{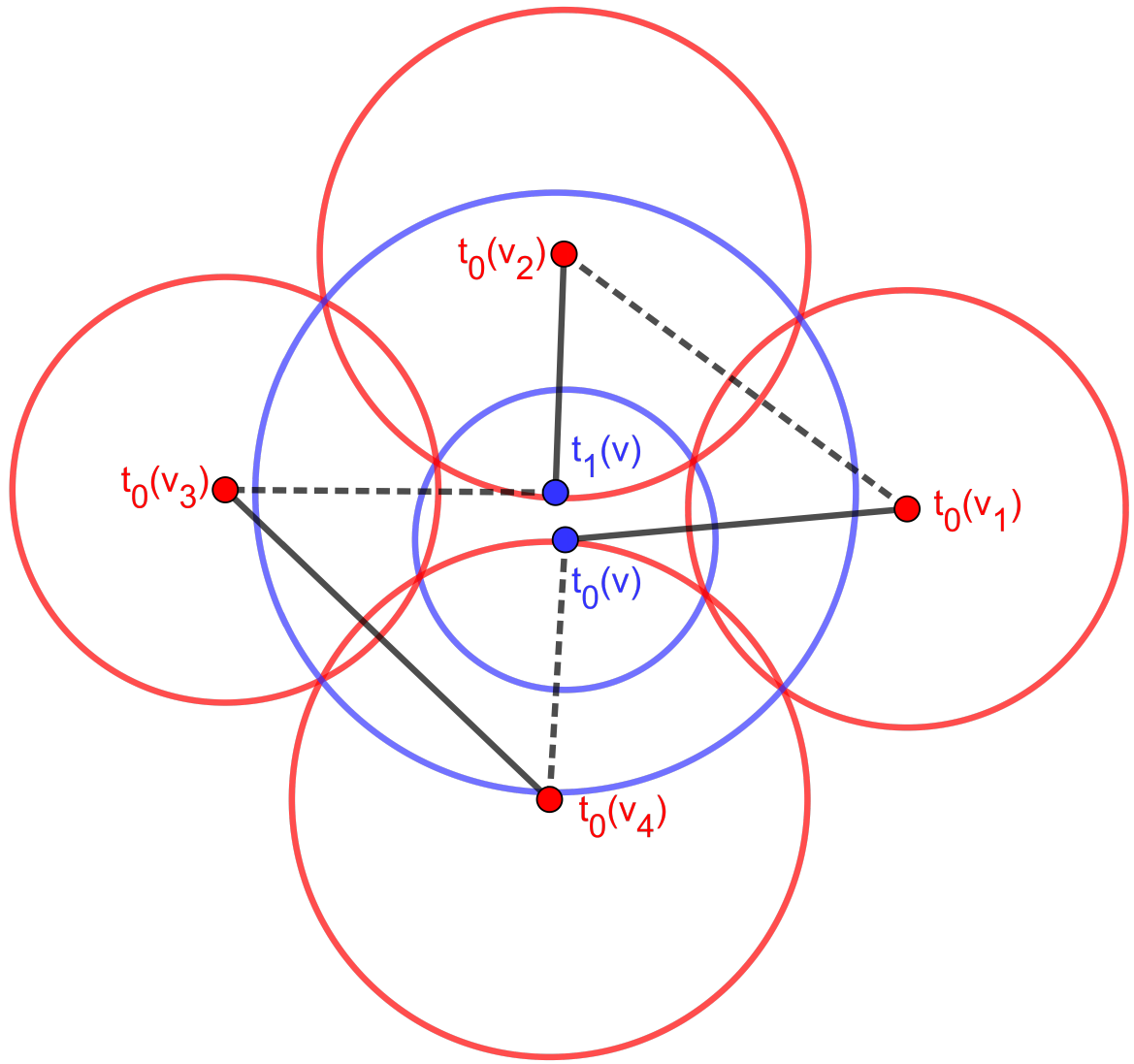}
	\caption{Local labeling in Miquel dynamics. The points involved in Equation~\eqref{eq:conicaldskp} for $k=0$ are highlighted on the right.}
	\label{fig:miquellabels}
\end{figure}

The \emph{multi-ratio} of an even number of points $a_1,a_2,\dots,a_{2m} \in \hC$ is defined as
\begin{align}
	&\mr(a_1,a_2,\dots,a_{2m}) \coloneqq \frac{(a_1-a_2)\cdots(a_{2m-1}-a_{2m})}{(a_2-a_3)\cdots(a_{2m}-a_1)} \in \hC.
\end{align}
Note that the multi-ratio is a rational function in any of the arguments and is therefore well-defined even if some of the arguments are $\infty$, in accordance with the rules of arithmetic of $\hC$. Moreover, it is also well-defined if non-consecutive arguments coincide, and if one pair of consecutive arguments coincides. However, if several consecutive arguments coincide then it is unclear how the multi-ratio should be defined.

In the following, we assume that the circle patterns involved are \emph{generic enough} that all the involved multi-ratios are well-defined. We assume this since we want to avoid lengthy case analyses of singularities, in particular since not every singularity is actually problematic for every result. For example, it is in principle no obstacle for Miquel dynamics if there is a circle of radius zero, but then the $X$-variables that we introduce later on are not well-defined. In other situations, Miquel dynamics may not be geometrically well-defined, but becomes well-defined if we enforce algebraic constraints, for example Equation~\eqref{eq:conicaldskp}. For instance, this happens when the centers of the six circles in Figure~\ref{fig:miquellabels} are on a line, which happens if $p_0(f_{23}) = p_0(f_{41})$ and $p_0(f_{12}) = p_0(f_{34})$. Note that some singularities of Miquel dynamics were studied in \cite{adtmdskp}.
Moreover, we assume throughout the paper that none of the circles of $c_k$ contains the point at infinity. This assumption is not necessary for the new Y-systems that we present, since they are defined using multi-ratios. However, the assumption is practical since we also revisit and compare with the old Y-system that involves circle centers. Note that the assumption implies that all the circle centers in $t$ are indeed in $\C$.

The following equation was discovered in \cite{amiquel} and independently in \cite{klrr}
\begin{align}
	\mr(t_k(v),t_k(v_1), t_k(v_2), t_{k+1}(v),t_k(v_3), t_k(v_4)) = -1, \label{eq:conicaldskp}
\end{align}
where the labeling of the faces is as shown in Figure~\ref{fig:miquellabels}. This equation is Equation~$(\chi_2)$ in the classification of integrable discrete equations of octahedron type \cite{abs}. Indeed, Equation~\eqref{eq:conicaldskp} posses the full symmetry of the octahedron, that is it is invariant under any cyclic permutation of the arguments and each argument swap $1 \leftrightarrow 4$, $2 \leftrightarrow 5$ and $3 \leftrightarrow 6$. As an equation on the octahedral lattice, Equation~\eqref{eq:conicaldskp} is also known as the \emph{dSKP equation}, studied in \cite{ksclifford} but first discovered in \cite{dndskp,ncwqdskp}.
We will explain the octahedral lattice in more detail in Section~\ref{sec:octahedralcombinatorics}, and the interpretation of Miquel dynamics on the octahedral lattice in Section~\ref{sec:miquelmaps}.

Consider a conical net $t_k$ and a vertex $v$ such that the center $t_k(v)$ is being replaced by Miquel dynamics in $t_{k+1}$. Then we define the \emph{$Y$-variables}
\begin{align}
	Y_k(v) = -\frac{(t_k(v_1) - t_k(v))(t_k(v_3) - t_k(v))}{(t_k(v_2) - t_k(v))(t_k(v_4) - t_k(v))}, \label{eq:yvars}
\end{align} 
where the labels are as in Figure~\ref{fig:miquellabels}. It was shown in \cite{amiquel,klrr} that $Y$ is real-valued. Moreover, if $t$ is a $t$-embedding then $Y$ is necessarily positive as well.

Let us call $-Y$ the negative $Y$-variables. It was also shown in \cite{amiquel,klrr} that the negative $Y$-variables together with Miquel dynamics constitute a \emph{Y-system} (see \cite{knsysystems}), that is, they satisfy
\begin{align}
	(-Y_{k+2}(v)) (-Y_k(v)) = \frac{(1 - (-Y_{k+1}(v_2)))(1 - (-Y_{k+1}(v_4)))}{(1 - (-Y_{k+1}^{-1}(v_1)))(1 - (-Y_{k+1}^{-1}(v_3)))}, \label{eq:ymin}
\end{align}
or equivalently
\begin{align}
	Y_{k+2}(v) Y_k(v) = \frac{(1 + Y_{k+1}(v_2))(1 + Y_{k+1}(v_4))}{(1 + Y_{k+1}^{-1}(v_1))(1 + Y_{k+1}^{-1}(v_3))}. \label{eq:ypos}
\end{align}

\begin{figure}
	\centering
	\includegraphics[scale=0.55]{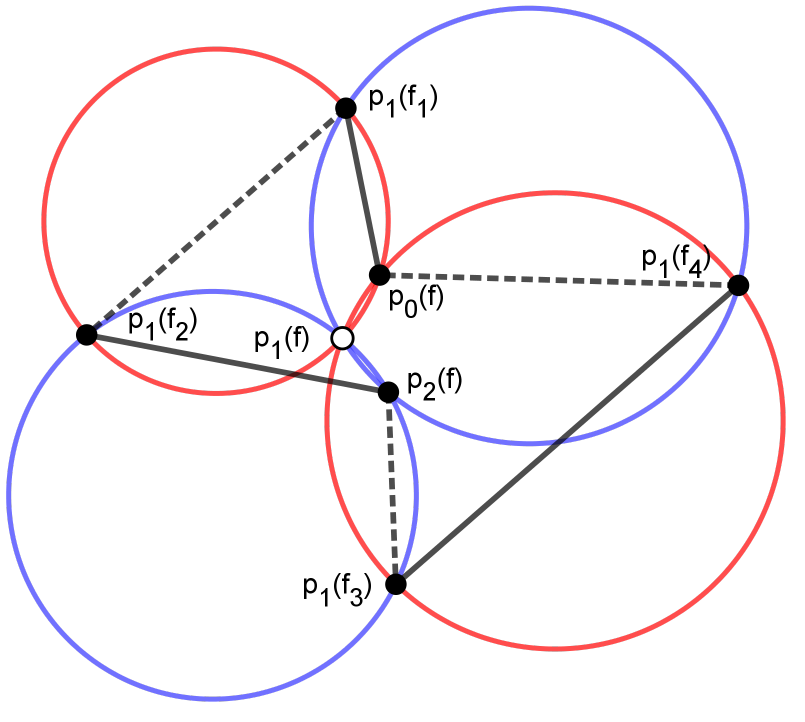}
	\hspace{3mm}
	\includegraphics[scale=0.4]{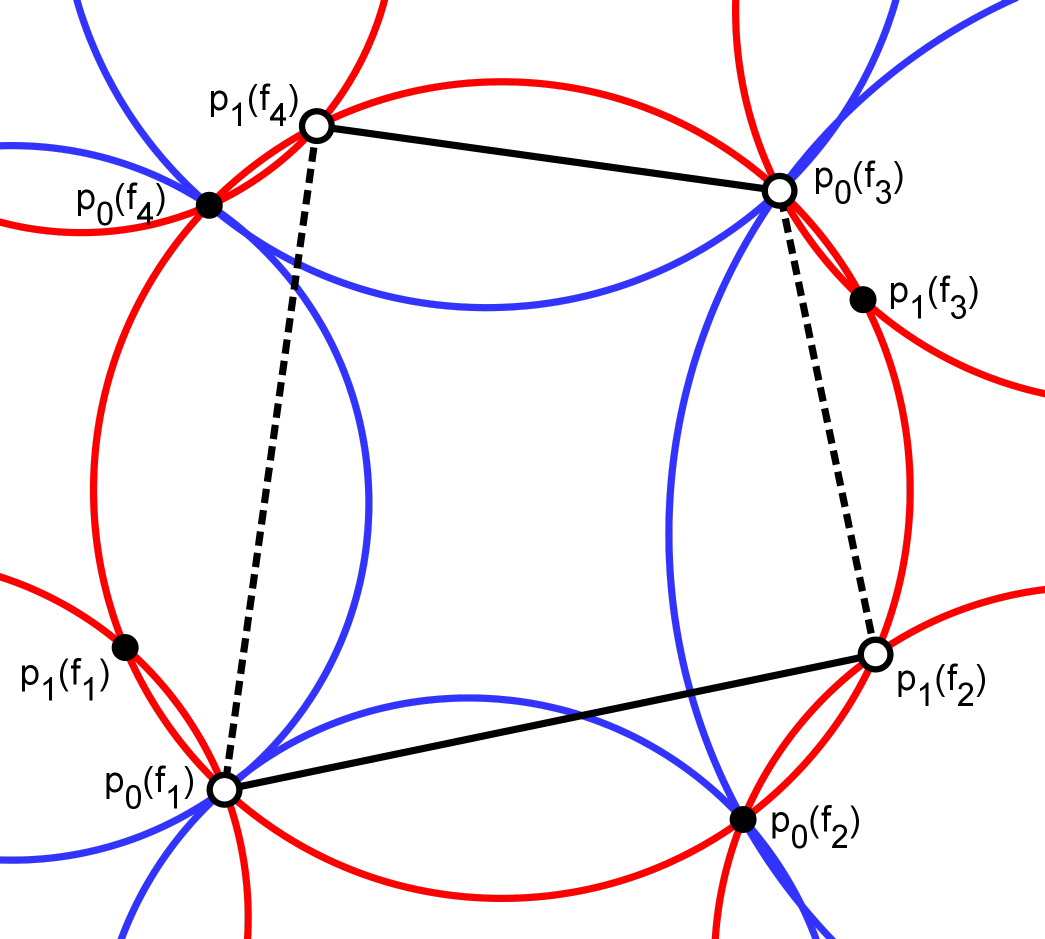}
	\caption{Left: a Clifford configuration, the red (resp.\ blue) circles are part of $c_0$ and $c_1$ (resp.\ $c_1$ and $c_2$). Right: the points that define the $\Xw$-variables, the red (resp.\ blue) circles are part of $c_{0}$ and $c_1$ (resp.\ $c_{-1}$ and $c_0$).}
	\label{fig:cliffordlabels}
\end{figure}

In the following, in a slight abuse of wording we call any variables that satisfy Equation~\eqref{eq:ypos} a Y-system. Therefore we will not need the notion of negative $Y$-variables again. The advantage of Equation~\eqref{eq:ypos} is that it is consistent with the previous Miquel dynamics literature \cite{amiquel, klrr} and the formulas for \emph{mutation} in \emph{cluster algebra} theory \cite{fzclusteralgebra}. This is also practical since the $Y$-variables determine the Boltzmann weights a statistical model, the dimer model, up to gauge and should therefore naturally be positive quantities \cite{amiquel,klrr}. Additionally, it is obvious from Equation~\eqref{eq:ypos} that if an initial circle pattern has positive $Y$-variables, the $Y$-variables remain positive under Miquel dynamics.

The first new result we present is that the dSKP equation does not only show up among the circle centers of a circle pattern, but also directly in the circle intersection points.
Specifically, we show in Theorem~\ref{th:udskp} that
\begin{align}
	\mr(p_{k-1}(f), p_k(f_1), p_k(f_2), p_{k+1}(f), p_k(f_3), p_k(f_4)) = -1, \label{eq:pdskp}
\end{align}
where the labels are as in Figure~\ref{fig:cliffordlabels} (left). Since $F(\Z^2) \simeq \Z^2$, the faces of $\Z^2$ may also be partitioned into even (parity) faces and odd (parity) faces. Let $p^\bullet_k$ be the restriction of $p_k$ to even faces if $k$ is even and the restriction to odd faces if $k$ is odd. Analogously, let $p^\circ_k$ be the restriction of $p_k$ to odd faces if $k$ is even and the restriction to even faces if $k$ is odd. Equation~\eqref{eq:pdskp} involves either only points of $\ub$ or only points of $\uw$. As a result, both $\ub$ and $\uw$ solve the dSKP equation on the octahedral lattice, as we show in Section~\ref{sec:dskp}.

Additionally, we show that there is a way to combine circle intersection points $p$ and circle centers $t$ that solves the dSKP equation on a higher-dimensional octahedral lattice, the $A_4$ lattice. This shows that from the dSKP perspective, (half) the intersection points are a Bäcklund transform of the circle centers and vice versa. In Section~\ref{sec:afour} we explain the combinatorics of $A_4$ and this result in more detail. It is beyond the scope of this paper, but in a similar vein one can show that if one adds additional boundary cluster variables to the $Y$-variables, the $X$-variables are related to these extended $Y$-variables by a sequence of cluster mutations, see \cite{athesis}. In particular, this suggests that the positivity of the $X$- and $Y$-variables may be more closely related than would be expected at first glance.

The multi-ratio $a_1, a_2, a_3, a_4 \in \hC$ of four points is called the \emph{cross-ratio}
\begin{align}
	\cro(a_1, a_2, a_3, a_4) = \mr(a_1, a_2, a_3, a_4).
\end{align}

Consider a circle $c_k(f)$ that is not replaced in the next iteration of Miquel dynamics and the labeling in Figure~\ref{fig:cliffordlabels} (right). We define the \emph{$\Xb$-variables} and the \emph{$\Xw$-variables}
\begin{align}
	X^\bullet_k(v) &= - \cro(p^\bullet_{k}(f_1), p^\bullet_{k+1}(f_2), p^\bullet_{k}(f_3), p^\bullet_{k+1}(f_4)), \label{eq:xfone}\\ 
	X^\circ_k(v) &= - \cro(p^\circ_{k+1}(f_1), p^\circ_{k}(f_2), p^\circ_{k+1}(f_3), p^\circ_{k}(f_4)).\label{eq:xftwo}
\end{align}
Since each $X_k(v)$ is the cross-ratio of four points on a circle, both the $\Xb$- and the $\Xw$-variables are real-valued. In Section~\ref{sec:toda} we show that both $\Xb$ and $\Xw$ constitute a Y-system. As a consequence, if the $X$-variables are positive initially, positivity is also preserved by Miquel dynamics.

A \emph{Möbius transformation} $M: \hC \rightarrow \hC$ is a transformation of the form
\begin{align}
	M(z) = \frac{az + b}{cz + d},
\end{align}
for some $a,b,c,d\in \C$ with $ad \neq bc$. The definition of a circle pattern is invariant under Möbius transformations of $\hC$. However, the Möbius transformation of a conical net is not a conical net, conical nets are only invariant under similarity transformations of $\hC$. Nor are the $Y$-variables invariant under Möbius transformations. On the other hand, cross-ratios are invariant under Möbius transformations and therefore the $X$-variables are also invariant under Möbius transformations. Therefore, unlike the $Y$-variables, the $X$-variables satisfy the transformation group principle formulated in \cite{bsorganizing,ddgbook} and mentioned in the introduction. There are several approaches to discrete conformality using circle patterns, for example using the \emph{length cross-ratio} \cite{bpsdiscretehyperbolic}, the circle intersection angles \cite{bsvariationalcp}, or cross-ratios of the form
\begin{align}
	 \cro(p^\bullet_{0}(f_1), p^\circ_{0}(f_2), p^\bullet_{0}(f_3), p^\circ_{0}(f_4)),
\end{align}
in \cite{bpdisosurfaces}. It would be interesting to understand how the $X$-variables relate to any of those theories, but it is currently unclear. However, since the $X$-variables have the right transformation behaviour, it is possible that a direct link exists.

\begin{table}
	\centering
	{\renewcommand{\arraystretch}{1.5}
		\begin{tabular}{c|c|c|c}
			& $t$-map & $\ub$-map & $\uw$-map \\ 
			\hline 
			$Y$-variables & real & complex & complex \\ 
			\hline 
			$X$-variables & complex & \textbf{real} & \textbf{real} \\ 
			\hline 
			$W$-variables & complex & \textbf{real iff ICP} &  \textbf{real iff ICP} \\ 
		\end{tabular}
	}
	\caption{Overview of the various Y-systems associated to Miquel dynamics, ICP is short for integrable circle pattern. In bold are the Y-systems that are invariant under Möbius transformations of the circle patterns.}
	\label{tab:vars}
\end{table}

Consider Equations \eqref{eq:xfone} and \eqref{eq:xftwo} but at a face $f$ that \emph{is} replaced by the next iteration of Miquel dynamics.
We call the resulting variables the $\Wb$- and $\Ww$-variables. In Section~\ref{sec:icp}, we show that these variables are also Y-systems, but not real-valued ones. Interestingly, they do satisfy $\Wb = \bar W^\circ$, that is they are complex conjugates. Therefore if the $\Wb$-variables are real, then $\Wb = \Ww$. We show that this special case of real $W$-variables corresponds to \emph{integrable circle patterns}, which were introduced in \cite{bmsanalytic}.

\begin{remark}
	Note that with the $Y$-, $X$- and $W$-variables there are three \emph{different} ways to obtain a Y-system for a dSKP map. The $Y$-variables are an \emph{affine} notion (affine geometry of $\C$), the $X$- and $W$-variables are a \emph{projective} notion (projective geometry of $\CP^1$), as explained in more detail in \cite{athesis}. It is also possible to associate projective variables to the conical net, and affine variables to the intersection points. However, in this case the resulting variables are not necessarily real-valued anymore, we give an overview in Table~\ref{tab:vars}. The affine variables for a dSKP map were first considered in \cite{amiquel,klrr}, but they already appear as cluster variables for \emph{T-graphs} in \cite{kenyonsheffield}, which are related to \emph{Menelaus lattices} as defined in \cite{ksclifford}. The projective variables for general dSKP maps are discussed in \cite{abs}, and coincide with the cluster variables for the pentagram map \cite{glickpentagram}.
\end{remark}

There are two more special cases that we mentioned in the introduction, which are motivated by statistical mechanics: harmonic embeddings and s-embeddings. Harmonic embeddings are circle patterns such that around every second face of $\Z^2$, the four circle centers $t(v_1)$, $t(v_2)$, $t(v_3)$, $t(v_4)$ form a rectangle. It was shown in \cite{klrr}, that in this case the $Y$-variables are in the so called \emph{resistor subvariety} (see \cite{gkdimers}), which means that
\begin{align}
	Y_0(v_1)Y_0(v_3) = Y_0(v_2) Y_0(v_4).
\end{align}
In Section~\ref{sec:harmonic}, we are able to show that $X^\bullet_k(v) = Y_k(v)$ for all vertices and $k\in\Z$. Apart from being a remarkable coincidence by itself, this also implies that the $X^\bullet_0$-variables are also in the resistor subvariety. The other type of special circle pattern is such that half of it is a circle packing, that is half the circles are touching (see Figure~\ref{fig:circlepacking}). In this case the corresponding conical net is called an s-embedding. In this case it was shown in \cite{klrr} that the $Y$-variables are in the so called \emph{Ising subvariety} (see \cite{kenyonpemantle}), which is equivalent to the fact that the $Y$-variables exhibit a global time-symmetry, in the sense that $Y_k(v) Y_{-k}(v) = 1$ for all vertices. For the $X$-variables, there is the same symmetry, but it instead relates $X^\bullet$-variables to $X^\circ$-variables in the sense that $X^\bullet_k(v) = X^\circ_{-k}(v)$.

\section{Combinatorics of $A_3$} \label{sec:octahedralcombinatorics}

\begin{figure}
	\centering
	\begin{tikzpicture}[scale=2,font=\sffamily]
	\coordinate (e1) at (1,0);
	\coordinate (e2) at (30:0.8);
	\coordinate (e3) at (0,1);
	
	\node[ivert] (v000) at ($0*(e1)+0*(e2)+0*(e3)$) {000};
	\node[wvert] (v100) at ($1*(e1)+0*(e2)+0*(e3)$) {};
	\node[wvert] (v010) at ($0*(e1)+1*(e2)+0*(e3)$) {};
	\node[wvert] (v001) at ($0*(e1)+0*(e2)+1*(e3)$) {};
	\node[ivert] (v110) at ($1*(e1)+1*(e2)+0*(e3)$) {110};
	\node[ivert] (v011) at ($0*(e1)+1*(e2)+1*(e3)$) {011};
	\node[ivert] (v101) at ($1*(e1)+0*(e2)+1*(e3)$) {101};
	\node[ivert] (v200) at ($2*(e1)+0*(e2)+0*(e3)$) {200};
	\node[ivert] (v020) at ($0*(e1)+2*(e2)+0*(e3)$) {020};
	\node[ivert] (v002) at ($0*(e1)+0*(e2)+2*(e3)$) {002};		
	\node[wvert] (v210) at ($2*(e1)+1*(e2)+0*(e3)$) {};
	\node[wvert] (v120) at ($1*(e1)+2*(e2)+0*(e3)$) {};
	\node[wvert] (v102) at ($1*(e1)+0*(e2)+2*(e3)$) {};		
	\node[wvert] (v201) at ($2*(e1)+0*(e2)+1*(e3)$) {};
	\node[wvert] (v021) at ($0*(e1)+2*(e2)+1*(e3)$) {};
	\node[wvert] (v012) at ($0*(e1)+1*(e2)+2*(e3)$) {};		
	\node[wvert] (v111) at ($1*(e1)+1*(e2)+1*(e3)$) {};		
	\node[ivert] (v211) at ($2*(e1)+1*(e2)+1*(e3)$) {211};		
	\node[ivert] (v121) at ($1*(e1)+2*(e2)+1*(e3)$) {121};		
	\node[ivert] (v112) at ($1*(e1)+1*(e2)+2*(e3)$) {112};		
	\node[ivert] (v220) at ($2*(e1)+2*(e2)+0*(e3)$) {220};
	\node[ivert] (v022) at ($0*(e1)+2*(e2)+2*(e3)$) {022};
	\node[ivert] (v202) at ($2*(e1)+0*(e2)+2*(e3)$) {202};
	\node[wvert] (v221) at ($2*(e1)+2*(e2)+1*(e3)$) {};		
	\node[wvert] (v212) at ($2*(e1)+1*(e2)+2*(e3)$) {};		
	\node[wvert] (v122) at ($1*(e1)+2*(e2)+2*(e3)$) {};		
	\node[ivert] (v222) at ($2*(e1)+2*(e2)+2*(e3)$) {222};

	\draw[-, line width=1pt]
	(v000) edge (v100) edge (v010) edge (v001)
	(v100) edge (v200) edge (v110) edge (v101)
	(v010) edge (v110) edge (v020) edge (v011)
	(v001) edge (v101) edge (v011) edge (v002)
	(v200) edge (v210) edge (v201)
	(v020) edge (v120) edge (v021) 
	(v002) edge (v102) edge (v012) 
	(v110) edge (v111) edge (v210) edge (v120)
	(v011) edge (v111) edge (v021) edge (v012)
	(v210) edge (v220) edge (v211)
	(v120) edge (v220) edge (v121)
	(v021) edge (v022) edge (v121)
	(v012) edge (v022) edge (v112)
	(v222) edge (v221) edge (v212) edge (v122)
	(v221) edge (v220) edge (v211) edge (v121)
	(v122) edge (v121) edge (v112) edge (v022)
	;
	\draw[-,white, line width=2.5pt]
	(v011) edge (v111)
	(v110) edge (v111)
	(v111) edge (v211) edge (v121) edge (v112)
	;
	\draw[-,black, line width=1pt]
	(v110) edge (v111)
	(v011) edge (v111)
	(v111) edge (v211) edge (v121) edge (v112)
	;
	\draw[-,white, line width=2.5pt]
	(v100) edge (v101)
	(v001) edge (v101) 
	(v101) edge (v111) edge (v201) edge (v102)
	(v200) edge (v201)
	(v002) edge (v102)
	(v201) edge (v202) edge (v211)
	(v210) edge (v211)
	(v102) edge (v202) edge (v112)
	(v012) edge (v112)
	(v212) edge (v211) edge (v202) edge (v112)
	;
	\draw[-,black, line width=1pt]
	(v100) edge (v101)
	(v001) edge (v101) 
	(v101) edge (v111) edge (v201) edge (v102)
	(v200) edge (v201)
	(v002) edge (v102)
	(v201) edge (v202) edge (v211)
	(v210) edge (v211)
	(v102) edge (v202) edge (v112)
	(v012) edge (v112)
	(v212) edge (v211) edge (v202) edge (v112)
	;
	
	\draw[fill=black,opacity=0.3,draw=none]
		(v000.center) -- (v100.center) -- (v110.center) -- (v111.center) -- (v011.center) -- (v001.center) -- cycle
		(v110.center) -- (v210.center) -- (v220.center) -- (v221.center) -- (v121.center) -- (v111.center) -- cycle
		(v011.center) -- (v111.center) -- (v121.center) -- (v122.center) -- (v022.center) -- (v012.center) -- cycle
	;
	\draw[fill=black,opacity=0.35,draw=none]
		(v101.center) -- (v201.center) -- (v211.center) -- (v212.center) -- (v112.center) -- (v102.center) -- cycle
	;
	
	\node[ivert] (v000) at ($0*(e1)+0*(e2)+0*(e3)$) {000};
	\node[wvert] (v100) at ($1*(e1)+0*(e2)+0*(e3)$) {};
	\node[wvert] (v010) at ($0*(e1)+1*(e2)+0*(e3)$) {};
	\node[wvert] (v001) at ($0*(e1)+0*(e2)+1*(e3)$) {};
	\node[ivert] (v110) at ($1*(e1)+1*(e2)+0*(e3)$) {110};
	\node[ivert] (v011) at ($0*(e1)+1*(e2)+1*(e3)$) {011};
	\node[ivert] (v101) at ($1*(e1)+0*(e2)+1*(e3)$) {101};
	\node[ivert] (v200) at ($2*(e1)+0*(e2)+0*(e3)$) {200};
	\node[ivert] (v020) at ($0*(e1)+2*(e2)+0*(e3)$) {020};
	\node[ivert] (v002) at ($0*(e1)+0*(e2)+2*(e3)$) {002};		
	\node[wvert] (v210) at ($2*(e1)+1*(e2)+0*(e3)$) {};
	\node[wvert] (v120) at ($1*(e1)+2*(e2)+0*(e3)$) {};
	\node[wvert] (v102) at ($1*(e1)+0*(e2)+2*(e3)$) {};		
	\node[wvert] (v201) at ($2*(e1)+0*(e2)+1*(e3)$) {};
	\node[wvert] (v021) at ($0*(e1)+2*(e2)+1*(e3)$) {};
	\node[wvert] (v012) at ($0*(e1)+1*(e2)+2*(e3)$) {};		
	\node[wvert] (v111) at ($1*(e1)+1*(e2)+1*(e3)$) {};		
	\node[ivert] (v211) at ($2*(e1)+1*(e2)+1*(e3)$) {211};		
	\node[ivert] (v121) at ($1*(e1)+2*(e2)+1*(e3)$) {121};		
	\node[ivert] (v112) at ($1*(e1)+1*(e2)+2*(e3)$) {112};		
	\node[ivert] (v220) at ($2*(e1)+2*(e2)+0*(e3)$) {220};
	\node[ivert] (v022) at ($0*(e1)+2*(e2)+2*(e3)$) {022};
	\node[ivert] (v202) at ($2*(e1)+0*(e2)+2*(e3)$) {202};
	\node[wvert] (v221) at ($2*(e1)+2*(e2)+1*(e3)$) {};		
	\node[wvert] (v212) at ($2*(e1)+1*(e2)+2*(e3)$) {};		
	\node[wvert] (v122) at ($1*(e1)+2*(e2)+2*(e3)$) {};		
	\node[ivert] (v222) at ($2*(e1)+2*(e2)+2*(e3)$) {222};
	\end{tikzpicture}
	\hspace{3mm}
	\begin{tikzpicture}[scale=2,font=\sffamily\scriptsize]
	
	\coordinate (e1) at (1,0);
	\coordinate (e2) at (30:0.9);
	\coordinate (e3) at (0,1);
	
	\node[bvert] (v000) at ($0*(e1)+0*(e2)+0*(e3)$) {};
	\node[wvert] (v100) at ($1*(e1)+0*(e2)+0*(e3)$) {};
	\node[wvert] (v010) at ($0*(e1)+1*(e2)+0*(e3)$) {};
	\node[wvert] (v001) at ($0*(e1)+0*(e2)+1*(e3)$) {};
	\node[bvert] (v110) at ($1*(e1)+1*(e2)+0*(e3)$) {};
	\node[bvert] (v011) at ($0*(e1)+1*(e2)+1*(e3)$) {};
	\node[bvert] (v101) at ($1*(e1)+0*(e2)+1*(e3)$) {};
	\node[bvert] (v200) at ($2*(e1)+0*(e2)+0*(e3)$) {};
	\node[bvert] (v020) at ($0*(e1)+2*(e2)+0*(e3)$) {};
	\node[bvert] (v002) at ($0*(e1)+0*(e2)+2*(e3)$) {};		
	\node[wvert] (v210) at ($2*(e1)+1*(e2)+0*(e3)$) {};
	\node[wvert] (v120) at ($1*(e1)+2*(e2)+0*(e3)$) {};
	\node[wvert] (v102) at ($1*(e1)+0*(e2)+2*(e3)$) {};		
	\node[wvert] (v201) at ($2*(e1)+0*(e2)+1*(e3)$) {};
	\node[wvert] (v021) at ($0*(e1)+2*(e2)+1*(e3)$) {};
	\node[wvert] (v012) at ($0*(e1)+1*(e2)+2*(e3)$) {};		
	\node[wvert] (v111) at ($1*(e1)+1*(e2)+1*(e3)$) {};		
	\node[bvert] (v211) at ($2*(e1)+1*(e2)+1*(e3)$) {};		
	\node[bvert] (v121) at ($1*(e1)+2*(e2)+1*(e3)$) {};		
	\node[bvert] (v112) at ($1*(e1)+1*(e2)+2*(e3)$) {};		
	\node[bvert] (v220) at ($2*(e1)+2*(e2)+0*(e3)$) {};
	\node[bvert] (v022) at ($0*(e1)+2*(e2)+2*(e3)$) {};
	\node[bvert] (v202) at ($2*(e1)+0*(e2)+2*(e3)$) {};
	\node[wvert] (v221) at ($2*(e1)+2*(e2)+1*(e3)$) {};		
	\node[wvert] (v212) at ($2*(e1)+1*(e2)+2*(e3)$) {};		
	\node[wvert] (v122) at ($1*(e1)+2*(e2)+2*(e3)$) {};		
	\node[bvert] (v222) at ($2*(e1)+2*(e2)+2*(e3)$) {};	
	
	\draw[-, gray]
		(v000) edge (v100) edge (v010) edge (v001)
		(v100) edge (v200) edge (v110) edge (v101)
		(v010) edge (v110) edge (v020) edge (v011)
		(v001) edge (v101) edge (v011) edge (v002)
		(v200) edge (v210) edge (v201)
		(v020) edge (v120) edge (v021) 
		(v002) edge (v102) edge (v012) 
		(v110) edge (v111) edge (v210) edge (v120)
		(v011) edge (v111) edge (v021) edge (v012)
		(v210) edge (v220) edge (v211)
		(v120) edge (v220) edge (v121)
		(v021) edge (v022) edge (v121)
		(v012) edge (v022) edge (v112)
		(v222) edge (v221) edge (v212) edge (v122)
		(v221) edge (v220) edge (v211) edge (v121)
		(v122) edge (v121) edge (v112) edge (v022)
	;
	\draw[-, blue, line width=1pt]
 		(v000) edge (v110) edge (v101) edge (v011)
 		(v110) edge (v220) edge[red] (v211) edge[red] (v121)
		(v220) -- (v211)  (v121) -- (v220)
		(v211) edge[red] (v121)
 		(v101) edge[red] (v211) edge[red] (v202) edge[red] (v112)
		(v211) -- (v202) -- (v112) edge[red] (v211)
 		(v011) edge[red] (v121) edge[red] (v112) edge (v022)
		(v121)  (v112) -- (v022) -- (v121) edge[red] (v112)
		(v200) edge (v110) edge (v101) edge (v211)
		(v002) edge (v112) edge (v101) edge (v011)
		(v020) edge (v110) edge (v121) edge (v011)
		(v222) edge (v112) edge (v121) edge (v211)
	;
	\draw[-,white, line width=2.5pt]
 		(v011) edge (v121) edge (v112) 
 		(v110) edge (v211) edge (v121)
	;
	\draw[-, blue, line width=1pt]
 		(v011) edge[red] (v121) edge[red] (v112) 
 		(v110) edge[red] (v211) edge[red] (v121)	
	;	
	\draw[-,white, line width=2.5pt]
		(v222) edge (v112) edge (v121) edge (v211)
 		(v110) -- (v101) -- (v011) -- (v110)
	;
	\draw[-, blue, line width=1pt]
		(v222) edge (v112) edge (v121) edge (v211)
	;
	\draw[-, red, line width=1pt]
 		(v110) -- (v101) -- (v011) -- (v110)	
	;
	\draw[-,white, line width=1.8pt]
		(v011) edge (v111)
		(v110) edge (v111)
		(v111) edge (v211) edge (v121) edge (v112)
	;
	\draw[-,gray]
		(v110) edge (v111)
		(v011) edge (v111)
		(v111) edge (v211) edge (v121) edge (v112)
	;
	\draw[-,white, line width=2.5pt]
 		(v101) edge (v211) edge (v202) edge (v112)	
		(v211) -- (v202) -- (v112) -- (v211)
		(v200) edge (v110) edge (v101) edge (v211)
		(v002) edge (v112) edge (v101) edge (v011)
 		(v000) edge (v101)
 	;
	\draw[-,blue, line width=1pt]
 		(v101) edge[red] (v211) edge (v202) edge[red] (v112)	
		(v211) -- (v202) -- (v112) edge[red] (v211)
		(v200) edge (v110) edge (v101) edge (v211)
		(v002) edge (v112) edge (v101) edge (v011)
 		(v000) edge (v101)
	;
	\draw[-,white, line width=1.8pt]
		(v100) edge (v101)
		(v001) edge (v101) 
		(v101) edge (v111) edge (v201) edge (v102)
		(v200) edge (v201)
		(v002) edge (v102)
		(v201) edge (v202) edge (v211)
		(v210) edge (v211)
		(v102) edge (v202) edge (v112)
		(v012) edge (v112)
		(v212) edge (v211) edge (v202) edge (v112)
		(v221) edge (v211)
		(v122) edge (v112)
	;
	\draw[-,gray]
		(v100) edge (v101)
		(v001) edge (v101) 
		(v101) edge (v111) edge (v201) edge (v102)
		(v200) edge (v201)
		(v002) edge (v102)
		(v201) edge (v202) edge (v211)
		(v210) edge (v211)
		(v102) edge (v202) edge (v112)
		(v012) edge (v112)
		(v212) edge (v211) edge (v202) edge (v112)
		(v221) edge (v211)
		(v122) edge (v112)
	;

	\draw[fill=black,opacity=0.3,draw=none]
		(v000.center) -- (v110.center) -- (v011.center) -- cycle
		(v011.center) -- (v121.center) -- (v022.center) -- cycle
		(v110.center) -- (v220.center) -- (v121.center) -- cycle	
	;
	\draw[fill=black,opacity=0.35,draw=none]
		(v101.center) -- (v211.center) -- (v112.center) -- cycle
	;
	
	\node[bvert] (v000) at ($0*(e1)+0*(e2)+0*(e3)$) {};
	\node[wvert] (v100) at ($1*(e1)+0*(e2)+0*(e3)$) {};
	\node[wvert] (v010) at ($0*(e1)+1*(e2)+0*(e3)$) {};
	\node[wvert] (v001) at ($0*(e1)+0*(e2)+1*(e3)$) {};
	\node[bvert] (v110) at ($1*(e1)+1*(e2)+0*(e3)$) {};
	\node[bvert] (v011) at ($0*(e1)+1*(e2)+1*(e3)$) {};
	\node[bvert] (v101) at ($1*(e1)+0*(e2)+1*(e3)$) {};
	\node[bvert] (v200) at ($2*(e1)+0*(e2)+0*(e3)$) {};
	\node[bvert] (v020) at ($0*(e1)+2*(e2)+0*(e3)$) {};
	\node[bvert] (v002) at ($0*(e1)+0*(e2)+2*(e3)$) {};		
	\node[wvert] (v210) at ($2*(e1)+1*(e2)+0*(e3)$) {};
	\node[wvert] (v120) at ($1*(e1)+2*(e2)+0*(e3)$) {};
	\node[wvert] (v102) at ($1*(e1)+0*(e2)+2*(e3)$) {};		
	\node[wvert] (v201) at ($2*(e1)+0*(e2)+1*(e3)$) {};
	\node[wvert] (v021) at ($0*(e1)+2*(e2)+1*(e3)$) {};
	\node[wvert] (v012) at ($0*(e1)+1*(e2)+2*(e3)$) {};		
	\node[wvert] (v111) at ($1*(e1)+1*(e2)+1*(e3)$) {};		
	\node[bvert] (v211) at ($2*(e1)+1*(e2)+1*(e3)$) {};		
	\node[bvert] (v121) at ($1*(e1)+2*(e2)+1*(e3)$) {};		
	\node[bvert] (v112) at ($1*(e1)+1*(e2)+2*(e3)$) {};		
	\node[bvert] (v220) at ($2*(e1)+2*(e2)+0*(e3)$) {};
	\node[bvert] (v022) at ($0*(e1)+2*(e2)+2*(e3)$) {};
	\node[bvert] (v202) at ($2*(e1)+0*(e2)+2*(e3)$) {};
	\node[wvert] (v221) at ($2*(e1)+2*(e2)+1*(e3)$) {};		
	\node[wvert] (v212) at ($2*(e1)+1*(e2)+2*(e3)$) {};		
	\node[wvert] (v122) at ($1*(e1)+2*(e2)+2*(e3)$) {};		
	\node[bvert] (v222) at ($2*(e1)+2*(e2)+2*(e3)$) {};		
	
	\end{tikzpicture}
	\caption{Left: $A_3$ identified with $\Z^3_+$. Right: edges of $A_3$, an octahedron is highlighted with red edges, and there are four black tetrahedra visible. Note that darker shaded areas are where two tetrahedra are overlapping.}
	\label{fig:athree}
\end{figure}

We need the subsets of $\Z^n$
\begin{align}
	\Z^n_c &= \{z \in \Z^{n} \ | \ \sum_{i=1}^{n} z_i = c \}, \\
	\Z^n_+ &= \{z \in \Z^{n} \ | \ \sum_{i=1}^{n} z_i \in 2 \Z \}, \\
	\Z^n_- &= \{z \in \Z^{n} \ | \ \sum_{i=1}^{n} z_i \in 2 \Z+1 \}.
\end{align}
With this notation, the \emph{$A_n$ lattice} is defined as $\Z^{n+1}_0$. For $A_3$ we use the representation as the \emph{octahedral lattice $\lat$}, where $\lat = \Z^3_+$. We will explain the bijection between $A_3$ and $\lat$ in Section~\ref{sec:afour}, for now we do not need it.

For $n\in \N$ we denote the index set $[n] = \{1,2,\dots, n\}$. For $i\in [n]$ let $e_i$ be the $i$-th unit vector. In order to improve the readability of equations, we introduce the \emph{shift operators} $\shi i$ and $\shi {\sm i}$.  For $i \in [n]$ and $z\in \Z^n$ the shift operators are defined by $\shi i z = z+ e_i$ and $\shi{\sm i} z = z - e_i$. Note that in the literature $\shi{\bar i}$ is sometimes used for $\shi{\sm i}$, but the former can be hard to read.
Additionally, if we give multiple numbers in the subscript this implies a composition of shift operators, for example $\shi{1,\sm 2} = \shi 1 \circ \shi{\sm 2}$.

Let us explain the 3-cells of $\lat$, see also Figure~\ref{fig:athree}. The \emph{octahedra} of $\lat$ are spanned by the points

\begin{align}
	\shi 1 z,\ \shim 1 z,\  \shi 2 z,\  \shim 2 z,\  \shi 3 z,\  \shim 3 z,
\end{align}
for all $z \in \Z^3_-$. We denote the set of octahedra by $\oct$ and identify $\oct \simeq \Z^3_-$.
The \emph{tetrahedra} of $\lat$ come in two kinds. The \emph{black tetrahedra} are spanned by the points 
\begin{align}
	z,\ \shi {2,3} z,\  \shi {1,3} z,\  \shi {1,2} z,
\end{align}
for all $z \in \Z^3_+$, and the \emph{white tetrahedra} are spanned by the points
\begin{align}
	\shi 1 z, \ \shi 2 z, \ \shi 3 z, \ \shi {1,2,3} z, 
\end{align}
for all $z \in \Z^3_-$. We denote the set of tetrahedra by $\tet \simeq \Z^3$, the set of black tetrahedra by $\tetb \simeq \Z^3_+$, and the set of white tetrahedra by $\tetw \simeq \Z^3_-$. It is sometimes useful to think of the set of tetrahedra as corresponding to the cubes of $\Z^3$.

Note that all the 2-cells of $\lat$ are triangles. We call two 3-cells adjacent if they share a triangle. We say a vertex $z\in \lat$ is incident to a 3-cell if $z$ is contained in the 3-cell.

\begin{remark}
	Miquel dynamics may also be defined as local dynamics if one works with more general combinatorics than $\Z^2$, replacing only one circle in each step, see \cite{amiquel,klrr}. In this case $\Z^2$ is replaced with a planar graph with bipartite dual. The definition of the $Y$-variables works analogously to Equation~\eqref{eq:yvars} and  dynamics are defined very naturally in the setup of cluster algebras. How to describe the $X$-variables for more general combinatorics is less clear. We believe the correct approach is to consider so called \emph{vector-relation configurations} \cite{agprvrc} defined in $\CP^1$ and such that the face weights are real. Miquel dynamics should then be described by the projective dynamics for vector-relation configurations, and the face weights should correspond to the $X$-variables.
\end{remark}

\section{Miquel maps} \label{sec:miquelmaps}

\begin{figure}
	\centering
	\begin{tikzpicture}[scale=2.2,font=\sffamily]
	\coordinate (e1) at (1,0);
	\coordinate (e2) at (60:0.7);
	\coordinate (e3) at (0,1);
	
	\node[ivert, blue] (v000) at ($0*(e1)+0*(e2)+0*(e3)$) {000};
	\node[wvert] (v100) at ($1*(e1)+0*(e2)+0*(e3)$) {};
	\node[wvert] (v010) at ($0*(e1)+1*(e2)+0*(e3)$) {};
	\node[wvert] (v001) at ($0*(e1)+0*(e2)+1*(e3)$) {};
	\node[ivert, blue] (v110) at ($1*(e1)+1*(e2)+0*(e3)$) {110};
	\node[ivert, red] (v011) at ($0*(e1)+1*(e2)+1*(e3)$) {011};
	\node[ivert, red] (v101) at ($1*(e1)+0*(e2)+1*(e3)$) {101};
	\node[ivert, blue] (v200) at ($2*(e1)+0*(e2)+0*(e3)$) {200};
	\node[ivert, blue] (v020) at ($0*(e1)+2*(e2)+0*(e3)$) {020};
	\node[wvert] (v210) at ($2*(e1)+1*(e2)+0*(e3)$) {};
	\node[wvert] (v120) at ($1*(e1)+2*(e2)+0*(e3)$) {};
	\node[wvert] (v201) at ($2*(e1)+0*(e2)+1*(e3)$) {};
	\node[wvert] (v021) at ($0*(e1)+2*(e2)+1*(e3)$) {};
	\node[wvert] (v111) at ($1*(e1)+1*(e2)+1*(e3)$) {};		
	\node[ivert, red] (v211) at ($2*(e1)+1*(e2)+1*(e3)$) {211};		
	\node[ivert, red] (v121) at ($1*(e1)+2*(e2)+1*(e3)$) {121};		
	\node[ivert, blue] (v220) at ($2*(e1)+2*(e2)+0*(e3)$) {220};
	\node[wvert] (v221) at ($2*(e1)+2*(e2)+1*(e3)$) {};		

	\node[ivert,rectangle,draw=black,inner sep=1.5pt] (c000) at ($.5*(e1)+.5*(e2)+.5*(e3)$) {000};	
	\node[ivert,rectangle,draw=black,inner sep=1.5pt] (c100) at ($1.5*(e1)+.5*(e2)+.5*(e3)$) {100};	
	\node[ivert,rectangle,draw=black,inner sep=1.5pt] (c010) at ($.5*(e1)+1.5*(e2)+.5*(e3)$) {010};	
	\node[ivert,rectangle,draw=black,inner sep=1.5pt] (c110) at ($1.5*(e1)+1.5*(e2)+.5*(e3)$) {110};	
	
	\draw[-, line width=1pt]
		(v000) edge (v100) edge (v010) edge (v001)
		(v100) edge (v200) edge (v110) edge (v101)
		(v010) edge (v110) edge (v020) edge (v011)
		(v001) edge (v101) edge (v011) 
		(v200) edge (v210) edge (v201)
		(v020) edge (v120) edge (v021) 
		(v110) edge (v111) edge (v210) edge (v120)
		(v011) edge (v111) edge (v021) 
		(v210) edge (v220) edge (v211)
		(v120) edge (v220) edge (v121)
		(v021) edge (v121)
		(v221) edge (v220) edge (v211) edge (v121)
	;
	\draw[-,white, line width=2.5pt]
		(v011) edge (v111)
		(v110) edge (v111)
		(v111) edge (v211) edge (v121)
	;
	\draw[-,black, line width=1pt]
		(v110) edge (v111)
		(v011) edge (v111)
		(v111) edge (v211) edge (v121)
	;
	\draw[-,white, line width=2.5pt]
		(v100) edge (v101)
		(v001) edge (v101) 
		(v101) edge (v111) edge (v201)
		(v200) edge (v201)
		(v201) edge (v211)
		(v210) edge (v211)
	;
	\draw[-,black, line width=1pt]
		(v100) edge (v101)
		(v001) edge (v101) 
		(v101) edge (v111) edge (v201)
		(v200) edge (v201)
		(v201) edge (v211)
		(v210) edge (v211)
	;
	
	\draw[fill=black,opacity=0.3,draw=none]
		(v000.center) -- (v100.center) -- (v110.center) -- (v111.center) -- (v011.center) -- (v001.center) -- cycle
		(v110.center) -- (v210.center) -- (v220.center) -- (v221.center) -- (v121.center) -- (v111.center) -- cycle
	;
	\draw[fill=black,opacity=0.35,draw=none]
	;
	
	\node[ivert, blue] (v000) at ($0*(e1)+0*(e2)+0*(e3)$) {000};
	\node[wvert] (v100) at ($1*(e1)+0*(e2)+0*(e3)$) {};
	\node[wvert] (v010) at ($0*(e1)+1*(e2)+0*(e3)$) {};
	\node[wvert] (v001) at ($0*(e1)+0*(e2)+1*(e3)$) {};
	\node[ivert, blue] (v110) at ($1*(e1)+1*(e2)+0*(e3)$) {110};
	\node[ivert, red] (v011) at ($0*(e1)+1*(e2)+1*(e3)$) {011};
	\node[ivert, red] (v101) at ($1*(e1)+0*(e2)+1*(e3)$) {101};
	\node[ivert, blue] (v200) at ($2*(e1)+0*(e2)+0*(e3)$) {200};
	\node[ivert, blue] (v020) at ($0*(e1)+2*(e2)+0*(e3)$) {020};
	\node[wvert] (v210) at ($2*(e1)+1*(e2)+0*(e3)$) {};
	\node[wvert] (v120) at ($1*(e1)+2*(e2)+0*(e3)$) {};
	\node[wvert] (v201) at ($2*(e1)+0*(e2)+1*(e3)$) {};
	\node[wvert] (v021) at ($0*(e1)+2*(e2)+1*(e3)$) {};
	\node[wvert] (v111) at ($1*(e1)+1*(e2)+1*(e3)$) {};		
	\node[ivert, red] (v211) at ($2*(e1)+1*(e2)+1*(e3)$) {211};		
	\node[ivert, red] (v121) at ($1*(e1)+2*(e2)+1*(e3)$) {121};		
	\node[ivert, blue] (v220) at ($2*(e1)+2*(e2)+0*(e3)$) {220};
	\node[wvert] (v221) at ($2*(e1)+2*(e2)+1*(e3)$) {};		

	\end{tikzpicture}
	\hspace{5mm}
	\includegraphics[scale=0.7]{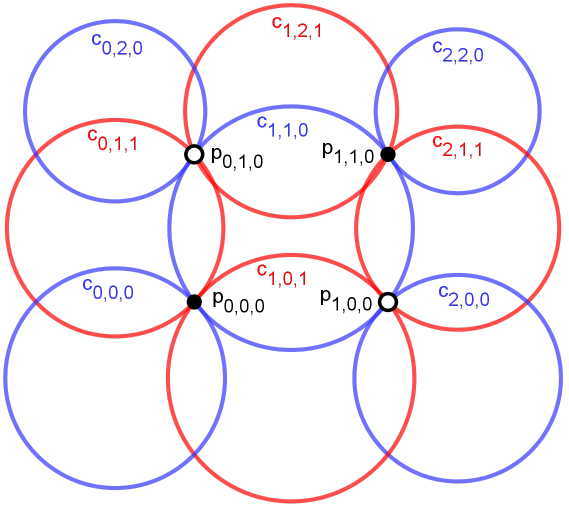}
	
	\caption{The $0$-layer of a Miquel map $(c,p)$, defined on vertices and tetrahedra of $\lat$ (left), and its image in $\C$ (right).}
	\label{fig:cpslice}
\end{figure}

In this section we explain how we view the circles and intersection points of Miquel dynamics on the octahedral lattice $\lat$.

\begin{definition}
	A \emph{Miquel map $(c,p)$} is a pair of maps 
	\begin{align}
		c&: \lat \rightarrow \mathrm{Circles}(\hC),\\
		p&: \tet \rightarrow \hC,
	\end{align}
	such that $p(T) \in c(z)$ for all incident $z \in \lat$, $T \in \tet$.
\end{definition}

Recall that we identified the set of tetrahedra $\tet$ with $\Z^3$ in Section~\ref{sec:octahedralcombinatorics}. For $k\in \Z$, consider the restriction $(c_k,p_k)$ of a Miquel map $(c,p)$ to
\begin{align}
	\{z \in \lat \ | \ z_3 \in\{k,k+1\} \} \cup \{T \in \tet \ | \ T_3 = k \}, \label{eq:miquellevel}
\end{align}
in the sense that
\begin{align}
	c_k(i,j) &= 
		\begin{cases}
			c(i,j,k) & i+j+k \in 2\Z, \\
			c(i,j,k+1) & i+j+k \in 2\Z+1, \\
		\end{cases} \\
	p_k(i,j) &= p(i,j,k).
\end{align}

Thus, $(c_k, p_k)$ is a circle pattern, see also Figure~\ref{fig:cpslice}. 

Recall that we identified the set of octahedra $\tet$ with $\Z^3_-$ in Section~\ref{sec:octahedralcombinatorics}. Consider an octahedron $O \in \oct$. Then we say that the six circles
\begin{align}
	 c(\shi 1 O),\  c(\shim 1 O),\  c(\shi 2  O),\  c( \shim2 O),\   c(\shi 3 O),\   c(\shim3 O),
\end{align}
of the vertices of $O$ constitute a \emph{Miquel configuration}, see also Figure~\ref{fig:miquelathree}. The eight intersection points 
\begin{align}
	p(\shi {\sm1,\sm2,\sm3} O), \ p(\shi {\sm1,\sm2}O), \ p(\shi {\sm1,\sm3} O), \ p(\shi {\sm2,\sm3} O), \ p(\shi {\sm1} O), \ p(\shi {\sm2} O), \ p(\shi {\sm3} O), \ p(O),
\end{align}
of the eight tetrahedra adjacent to $O$ are the intersection points of the Miquel configuration.

Each point $z \in \lat$ with $z_3 = k+2$ belongs to exactly one octahedron $O$ with $O_3 = k+1$, specifically $O = z - e_3$. Therefore each circle in $(c_{k+1}, p_{k+1})$ either coincides with a circle in $(c_k, p_k)$ or is determined via a Miquel configuration from $(c_k, p_k)$. Moreover, each tetrahedron $T \in \tet$ with $T_3 = k+1$ is adjacent to exactly one octahedron $O$ with $O_3 = k+1$. Therefore also all the intersection points of $(c_{k+1}, p_{k+1})$ are determined via Miquel configurations. Also note that for an octahedron $O \in \oct \simeq \Z^3_-$ the point $\shi{3} O \in \Z^3$ is not an octahedron. Therefore if a circle is in $(c_{k+1}, p_{k+1})$ but not in $(c_{k}, p_{k})$, it must also be in $(c_{k+2}, p_{k+2})$. We conclude that $(c_{k+1}, p_{k+1})$ is obtained from $(c_k, p_k)$ via Miquel dynamics. Therefore Miquel maps are the lattice perspective on (the iteration of) Miquel dynamics.

\begin{figure}
	\centering
	\includegraphics[scale=0.45]{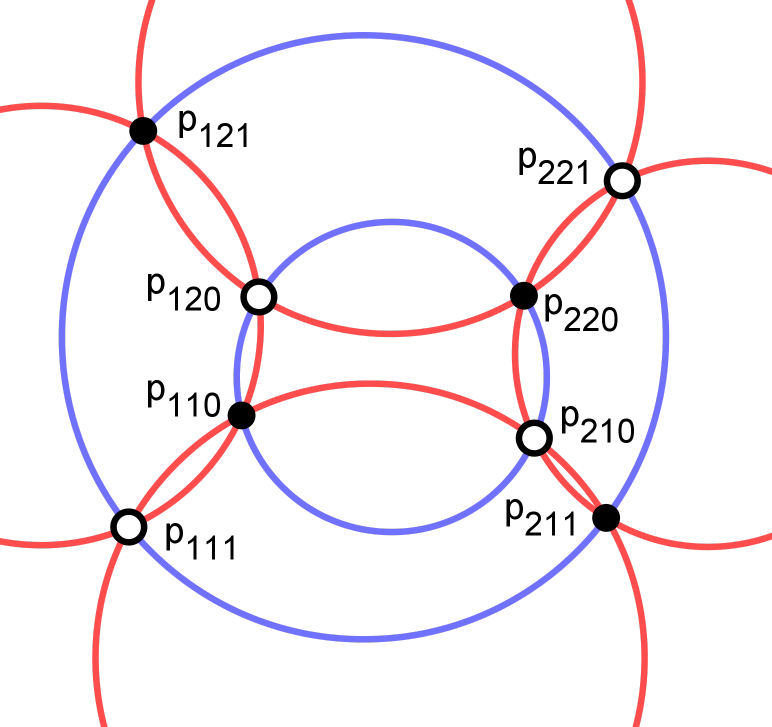}
	\hspace{4mm}
	\includegraphics[scale=0.45]{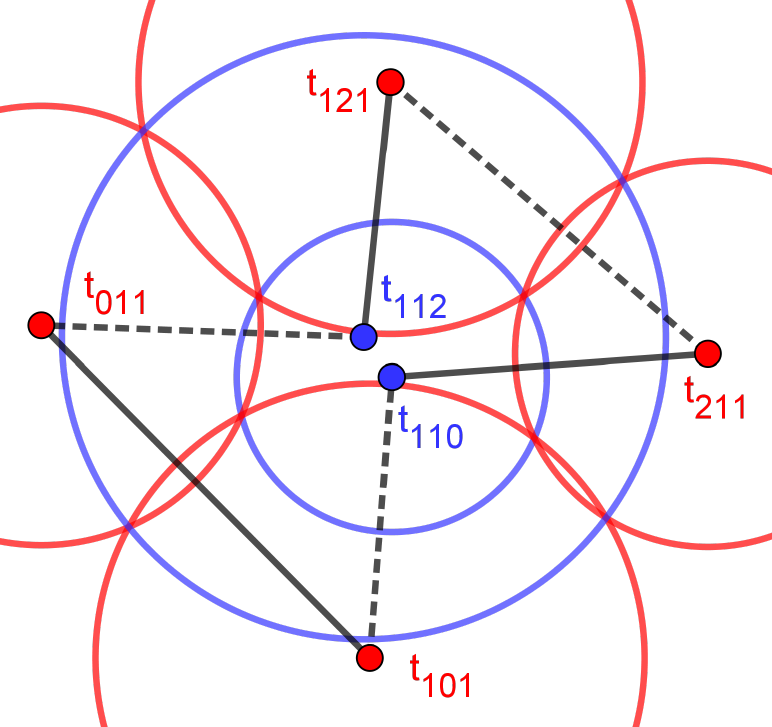}
	\caption{Coordinates in the octahedral lattice for a Miquel configuration at octahedron $(1,1,1)$. }
	\label{fig:miquelathree}
\end{figure}

\section{dSKP maps} \label{sec:dskp}

\begin{definition}\label{def:dskp}
	A \emph{dSKP map} $a: \Z^3_\pm \rightarrow \hC$ is a map that satisfies the \emph{dSKP lattice equation}, that is a map such that 
	\begin{align}
		\mr(a(\shi 1 z), a(\shi 2 z), a(\shi 3 z), a(\shi{\sm 1} z), a(\shi {\sm 2} z), a(\shi {\sm 3}  z)) = -1. \label{eq:dskp}
	\end{align}
	for all $z \in \Z_{\mp}$.
\end{definition}

Assume we know the values of a dSKP map $f$ restricted to $\{z \in \Z^3_{\pm} \ | \ z_3 \in \{0,1\}  \}$, and let us call this restriction \emph{initial data}. Then every value of $f(z)$ for $z\in \Z^3_{\pm}$ with $z_3 = 2$ occurs in exactly one instance of Equation~\eqref{eq:dskp}. Hence, the restriction of $f$ to $\{z \in \Z^3_{\pm} \ | \ z_3 = 2 \}$ is determined by the initial data. Since this argument may be iterated (forwards and backwards), all of the dSKP map $f$ is determined by the initial data.

Before we turn to new results, let us quickly recapitulate a previously discovered result relating Miquel maps to dSKP maps.
Given a Miquel map $(c,p)$ we associate a \emph{$t$-map} $t: \lat \rightarrow \hC$, which consists of the circle centers of $c$, that is for $z \in \lat$ the point $t(z)$ is the center of $c(z)$.
The following theorem was found and proven in \cite{amiquel} and independently in \cite{klrr}.
\begin{theorem}\label{th:tdskp}
	The $t$-map of a Miquel map is a dSKP map.
\end{theorem}

Consequently, the centers $t_0$ of a circle pattern $(c_0, p_0)$ are initial data and determine the whole $t$-map.

Let us now turn to an analogous new result. For this, we need to split the intersection points of a Miquel map according to their parity on $\Z^3$.

\begin{definition}
	Let $(c,p)$ be a Miquel map. The $\ub$-map is the map $\ub: \tetb \rightarrow \hC$-map which is the restriction of $p$ to $\tetb$. The $\uw$-map is the map $\uw: \tetw \rightarrow \hC$-map which is the restriction of $p$ to $\tetw$. 
\end{definition}

\begin{figure}
	\centering
	\includegraphics[scale=0.5]{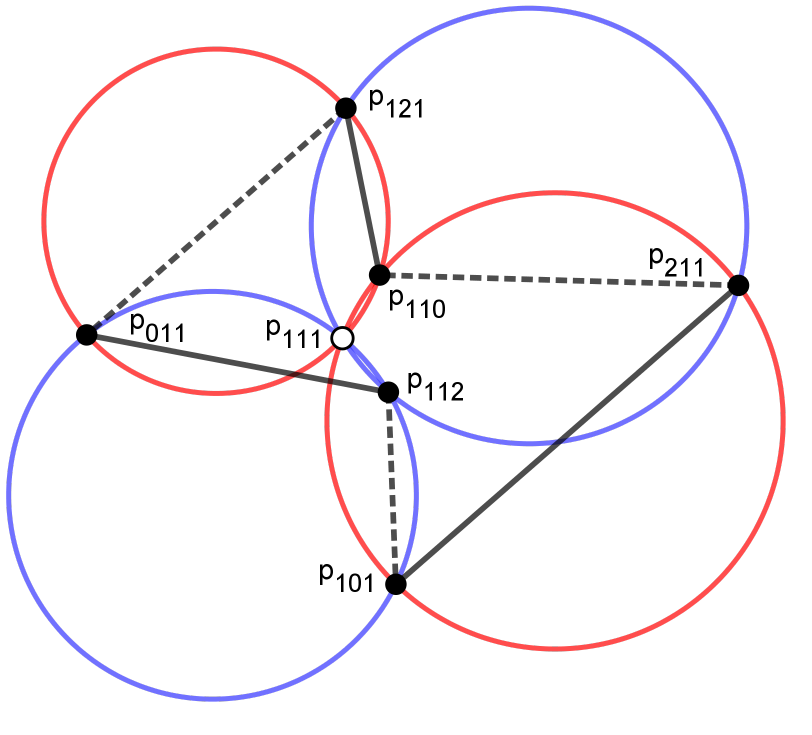}
	\hspace{5mm}
	\includegraphics[scale=0.5]{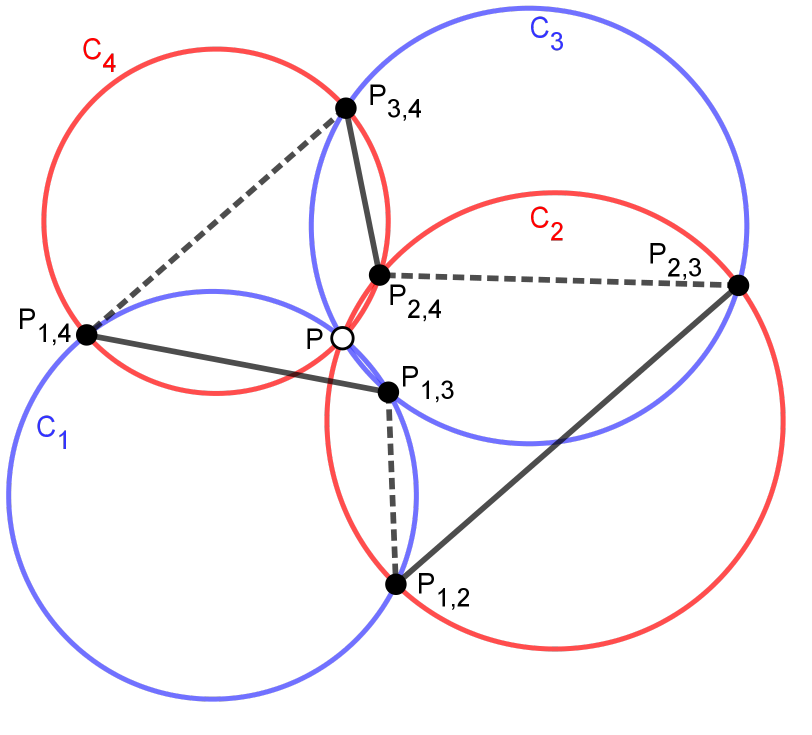}
	\caption{Left: coordinates in the octahedral lattice for a Clifford configuration at white tetrahedron $(1,1,1)$, the dSKP equation for $\ub$ is highlighted. Right: labeling for the proof of \ref{th:udskp}.}
	\label{fig:cpathree}
\end{figure}

\begin{theorem} \label{th:udskp}
	The $\ub$-map and $\uw$-map of a Miquel map are dSKP maps.
\end{theorem}

\begin{proof}
	It has been shown in \cite{ksclifford} how the dSKP equation arises in a Clifford 4-circle configuration. Let us recall this result. Consider four circles $C_1, C_2, C_3, C_4$ that intersect in a point $P$, and let $P_{i,j}$ be the other intersection point of $P_i$ and $P_j$ for $i,j\in\{1,2,3,4\}$, see also Figure~\ref{fig:cpathree}. Then it was shown that
	\begin{align}
		\mr(P_{2,3},P_{3,4}, P_{2,4}, P_{1,4}, P_{1,2},P_{1,3}) = -1.
	\end{align}
	Now let us apply this result to our claim. Let $z \in \lat$ and consider the identification of circles
	\begin{align}
		c(z) = C_1, \  c(\shi{1,3} z) = C_2,\ c(\shi{1,2} z) = C_3, \ c(\shi{2,3} z) = C_4.
	\end{align}
	All four circles pass through $p(z) = \uw(z)$, so we may identify $\uw(z) = P$. The remaining point identifications are
	\begin{align}
		\ub(\shi{\sm 2}z) &= P_{1,2},&  \ub(\shi{1}z) &= P_{2,3},&  \ub(\shi{2}z) &= P_{3,4},\\
		 \ub(\shi{\sm 1}z) &= P_{1,4} & \ub(\shi{\sm 3}z) &= P_{1,3} & \ub(\shi{3}z) &= P_{2,4}.
	\end{align}
	Therefore, the map $\ub$ satisfies
	\begin{align}
		\mr(\ub(\shi 1  z), \ub(\shi 2  z),\ub(\shi 3  z),\ub(\shi{\sm 1}  z),\ub(\shi {\sm 2}  z), \ub(\shi {\sm 3} z)) = -1, \label{eq:udskp}
	\end{align}	
	as required by Definition~\ref{def:dskp}. An analogous argument works for $\uw$.
\end{proof}

As a consequence of Theorem~\ref{th:udskp}, the restriction of the $\ub$-map to $\{z \in \Z^3_{\pm} \ | \ z_3 \in \{0,1\} \}$ forms initial data and therefore determines the whole $\ub$-map, and analogously for the $\uw$-map.

\begin{remark}
	In \cite{adtmdskp} it was shown how values of the $t$-map $t$ can be expressed in terms of initial data $t^0, t^1$ using dimer partition functions. The same tools apply to express the values of $\ub$ and $\uw$ in terms of initial data.
\end{remark}

\section{A combined dSKP map on $A_4$} \label{sec:afour}

We need some combinatorial and geometric preliminaries. We begin with combinatorics.
Let us recall the definition of the $A_4$ lattice, namely
\begin{align}
	A_4 = \Z^5_0 = \{ z \in \Z^5 \ | \ \sum_{i=1}^5 z_i = 0 \}.
\end{align}
In the following, we consider the subset
\begin{align}
	\hat A_4 = \left\{ z \in A_4 \ | \ z_5 \in \{-1,0,1\} \right\}.
\end{align}
Next, we introduce an identification $\phi:  \hlat \rightarrow \hat A_4$ of $\hlat = \lat \cup \tet$ with $\hat A_4$ as follows, see also Figure~\ref{fig:afour}. Let us identify $\tet$ with the cubes of $\Z^3$, which in turn we identify with $\Z^3 + \frac12(1,1,1)$. Introduce the four vectors in $\Z^5$
\begin{align}
	a_1  &= (1,0,0,0,-1), & a_2 &= (0,1,0,0,-1),\\
	a_3 &= (0,0,1,0,-1), & a_4 &= (0,0,0,1,-1),
\end{align}
and the four vectors in $\frac12\Z^3$
\begin{align}
	b_1  &= \frac12(+1,-1,-1), & b_2 &= \frac12(-1,+1,-1),\\
	b_3 &= \frac12(-1,-1,+1), & b_4 &= 	\frac12(+1,+1,+1).
\end{align}

\begin{figure}
	\centering
	\begin{tikzpicture}[scale=2,font=\sffamily]
		\coordinate (e1) at (1,0);
		\coordinate (e2) at (30:0.8);
		\coordinate (e3) at (0,1);
	
		\node[ivert] (v000) at ($0*(e1)+0*(e2)+0*(e3)$) {123\={4}};
		\node[wvert] (v100) at ($1*(e1)+0*(e2)+0*(e3)$) {};
		\node[wvert] (v010) at ($0*(e1)+1*(e2)+0*(e3)$) {};
		\node[wvert] (v001) at ($0*(e1)+0*(e2)+1*(e3)$) {};
		\node[ivert] (v110) at ($1*(e1)+1*(e2)+0*(e3)$) {12};
		\node[ivert] (v011) at ($0*(e1)+1*(e2)+1*(e3)$) {23};
		\node[ivert] (v101) at ($1*(e1)+0*(e2)+1*(e3)$) {13};
		\node[ivert] (v200) at ($2*(e1)+0*(e2)+0*(e3)$) {11};
		\node[ivert] (v020) at ($0*(e1)+2*(e2)+0*(e3)$) {22};
		\node[ivert] (v002) at ($0*(e1)+0*(e2)+2*(e3)$) {33};		
		\node[wvert] (v210) at ($2*(e1)+1*(e2)+0*(e3)$) {};
		\node[wvert] (v120) at ($1*(e1)+2*(e2)+0*(e3)$) {};
		\node[wvert] (v102) at ($1*(e1)+0*(e2)+2*(e3)$) {};		
		\node[wvert] (v201) at ($2*(e1)+0*(e2)+1*(e3)$) {};
		\node[wvert] (v021) at ($0*(e1)+2*(e2)+1*(e3)$) {};
		\node[wvert] (v012) at ($0*(e1)+1*(e2)+2*(e3)$) {};		
		\node[wvert] (v111) at ($1*(e1)+1*(e2)+1*(e3)$) {};		
		\node[ivert] (v211) at ($2*(e1)+1*(e2)+1*(e3)$) {14};		
		\node[ivert] (v121) at ($1*(e1)+2*(e2)+1*(e3)$) {24};		
		\node[ivert] (v112) at ($1*(e1)+1*(e2)+2*(e3)$) {34};		
		\node[ivert] (v220) at ($2*(e1)+2*(e2)+0*(e3)$) {12\={3}4};
		\node[ivert] (v022) at ($0*(e1)+2*(e2)+2*(e3)$) {\={1}234};
		\node[ivert] (v202) at ($2*(e1)+0*(e2)+2*(e3)$) {1\={2}34};
		\node[wvert] (v221) at ($2*(e1)+2*(e2)+1*(e3)$) {};		
		\node[wvert] (v212) at ($2*(e1)+1*(e2)+2*(e3)$) {};		
		\node[wvert] (v122) at ($1*(e1)+2*(e2)+2*(e3)$) {};		
		\node[ivert] (v222) at ($2*(e1)+2*(e2)+2*(e3)$) {44};

		\draw[-, line width=1pt]
			(v000) edge (v100) edge (v010) edge (v001)
			(v100) edge (v200) edge (v110) edge (v101)
			(v010) edge (v110) edge (v020) edge (v011)
			(v001) edge (v101) edge (v011) edge (v002)
			(v200) edge (v210) edge (v201)
			(v020) edge (v120) edge (v021) 
			(v002) edge (v102) edge (v012) 
			(v110) edge (v111) edge (v210) edge (v120)
			(v011) edge (v111) edge (v021) edge (v012)
			(v210) edge (v220) edge (v211)
			(v120) edge (v220) edge (v121)
			(v021) edge (v022) edge (v121)
			(v012) edge (v022) edge (v112)
			(v222) edge (v221) edge (v212) edge (v122)
			(v221) edge (v220) edge (v211) edge (v121)
			(v122) edge (v121) edge (v112) edge (v022)
		;
		\draw[-,white, line width=2.5pt]
			(v011) edge (v111)
			(v110) edge (v111)
			(v111) edge (v211) edge (v121) edge (v112)
		;
		\draw[-,black, line width=1pt]
			(v110) edge (v111)
			(v011) edge (v111)
			(v111) edge (v211) edge (v121) edge (v112)
		;
		\draw[-,white, line width=2.5pt]
			(v100) edge (v101)
			(v001) edge (v101) 
			(v101) edge (v111) edge (v201) edge (v102)
			(v200) edge (v201)
			(v002) edge (v102)
			(v201) edge (v202) edge (v211)
			(v210) edge (v211)
			(v102) edge (v202) edge (v112)
			(v012) edge (v112)
			(v212) edge (v211) edge (v202) edge (v112)
		;
		\draw[-,black, line width=1pt]
			(v100) edge (v101)
			(v001) edge (v101) 
			(v101) edge (v111) edge (v201) edge (v102)
			(v200) edge (v201)
			(v002) edge (v102)
			(v201) edge (v202) edge (v211)
			(v210) edge (v211)
			(v102) edge (v202) edge (v112)
			(v012) edge (v112)
			(v212) edge (v211) edge (v202) edge (v112)
		;

		\draw[fill=black,opacity=0.3,draw=none]
			(v000.center) -- (v100.center) -- (v110.center) -- (v111.center) -- (v011.center) -- (v001.center) -- cycle
			(v101.center) -- (v201.center) -- (v211.center) -- (v212.center) -- (v112.center) -- (v102.center) -- cycle
			(v110.center) -- (v210.center) -- (v220.center) -- (v221.center) -- (v121.center) -- (v111.center) -- cycle
			(v011.center) -- (v111.center) -- (v121.center) -- (v122.center) -- (v022.center) -- (v012.center) -- cycle
		;
		\draw[fill=black,opacity=0.35,draw=none]
			(v101.center) -- (v201.center) -- (v211.center) -- (v212.center) -- (v112.center) -- (v102.center) -- cycle
		;

		\node[ivert] (v000) at ($0*(e1)+0*(e2)+0*(e3)$) {123\={4}};
		\node[wvert] (v100) at ($1*(e1)+0*(e2)+0*(e3)$) {};
		\node[wvert] (v010) at ($0*(e1)+1*(e2)+0*(e3)$) {};
		\node[wvert] (v001) at ($0*(e1)+0*(e2)+1*(e3)$) {};
		\node[ivert] (v110) at ($1*(e1)+1*(e2)+0*(e3)$) {12};
		\node[ivert] (v011) at ($0*(e1)+1*(e2)+1*(e3)$) {23};
		\node[ivert] (v101) at ($1*(e1)+0*(e2)+1*(e3)$) {13};
		\node[ivert] (v200) at ($2*(e1)+0*(e2)+0*(e3)$) {11};
		\node[ivert] (v020) at ($0*(e1)+2*(e2)+0*(e3)$) {22};
		\node[ivert] (v002) at ($0*(e1)+0*(e2)+2*(e3)$) {33};		
		\node[wvert] (v210) at ($2*(e1)+1*(e2)+0*(e3)$) {};
		\node[wvert] (v120) at ($1*(e1)+2*(e2)+0*(e3)$) {};
		\node[wvert] (v102) at ($1*(e1)+0*(e2)+2*(e3)$) {};		
		\node[wvert] (v201) at ($2*(e1)+0*(e2)+1*(e3)$) {};
		\node[wvert] (v021) at ($0*(e1)+2*(e2)+1*(e3)$) {};
		\node[wvert] (v012) at ($0*(e1)+1*(e2)+2*(e3)$) {};		
		\node[wvert] (v111) at ($1*(e1)+1*(e2)+1*(e3)$) {};		
		\node[ivert] (v211) at ($2*(e1)+1*(e2)+1*(e3)$) {14};		
		\node[ivert] (v121) at ($1*(e1)+2*(e2)+1*(e3)$) {24};		
		\node[ivert] (v112) at ($1*(e1)+1*(e2)+2*(e3)$) {34};		
		\node[ivert] (v220) at ($2*(e1)+2*(e2)+0*(e3)$) {12\={3}4};
		\node[ivert] (v022) at ($0*(e1)+2*(e2)+2*(e3)$) {\={1}234};
		\node[ivert] (v202) at ($2*(e1)+0*(e2)+2*(e3)$) {1\={2}34};
		\node[wvert] (v221) at ($2*(e1)+2*(e2)+1*(e3)$) {};		
		\node[wvert] (v212) at ($2*(e1)+1*(e2)+2*(e3)$) {};		
		\node[wvert] (v122) at ($1*(e1)+2*(e2)+2*(e3)$) {};		
		\node[ivert] (v222) at ($2*(e1)+2*(e2)+2*(e3)$) {44};
	\end{tikzpicture}
	\hspace{2mm}
	\begin{tikzpicture}[scale=2,font=\sffamily\scriptsize]
		
		\coordinate (e1) at (1,0);
		\coordinate (e2) at (30:0.8);
		\coordinate (e3) at (0,1);
		
		\node[ivert] (v000) at ($0*(e1)+0*(e2)+0*(e3)$) {123\={4}};
		\coordinate (v100) at ($1*(e1)+0*(e2)+0*(e3)$) {};
		\coordinate (v010) at ($0*(e1)+1*(e2)+0*(e3)$) {};
		\coordinate (v001) at ($0*(e1)+0*(e2)+1*(e3)$) {};
		\node[ivert] (v110) at ($1*(e1)+1*(e2)+0*(e3)$) {12};
		\node[ivert] (v011) at ($0*(e1)+1*(e2)+1*(e3)$) {23};
		\node[ivert] (v101) at ($1*(e1)+0*(e2)+1*(e3)$) {13};
		\node[ivert] (v200) at ($2*(e1)+0*(e2)+0*(e3)$) {11};
		\node[ivert] (v020) at ($0*(e1)+2*(e2)+0*(e3)$) {22};
		\node[ivert] (v002) at ($0*(e1)+0*(e2)+2*(e3)$) {33};		
		\coordinate (v210) at ($2*(e1)+1*(e2)+0*(e3)$) {};
		\coordinate (v120) at ($1*(e1)+2*(e2)+0*(e3)$) {};
		\coordinate (v102) at ($1*(e1)+0*(e2)+2*(e3)$) {};		
		\coordinate (v201) at ($2*(e1)+0*(e2)+1*(e3)$) {};
		\coordinate (v021) at ($0*(e1)+2*(e2)+1*(e3)$) {};
		\coordinate (v012) at ($0*(e1)+1*(e2)+2*(e3)$) {};		
		\coordinate (v111) at ($1*(e1)+1*(e2)+1*(e3)$) {};		
		\node[ivert] (v211) at ($2*(e1)+1*(e2)+1*(e3)$) {14};		
		\node[ivert] (v121) at ($1*(e1)+2*(e2)+1*(e3)$) {24};		
		\node[ivert] (v112) at ($1*(e1)+1*(e2)+2*(e3)$) {34};		
		\node[ivert] (v220) at ($2*(e1)+2*(e2)+0*(e3)$) {12\={3}4};
		\node[ivert] (v022) at ($0*(e1)+2*(e2)+2*(e3)$) {\={1}234};
		\node[ivert] (v202) at ($2*(e1)+0*(e2)+2*(e3)$) {1\={2}34};
		\coordinate (v221) at ($2*(e1)+2*(e2)+1*(e3)$) {};		
		\coordinate (v212) at ($2*(e1)+1*(e2)+2*(e3)$) {};		
		\coordinate (v122) at ($1*(e1)+2*(e2)+2*(e3)$) {};		
		\node[ivert] (v222) at ($2*(e1)+2*(e2)+2*(e3)$) {44};		
		
		\node[ivert] (c000) at ($.5*(e1)+.5*(e2)+.5*(e3)$) {123};	
		\node[ivert] (c100) at ($1.5*(e1)+.5*(e2)+.5*(e3)$) {1};	
		\node[ivert] (c010) at ($.5*(e1)+1.5*(e2)+.5*(e3)$) {2};	
		\node[ivert] (c001) at ($.5*(e1)+.5*(e2)+1.5*(e3)$) {3};	
		\node[ivert] (c110) at ($1.5*(e1)+1.5*(e2)+.5*(e3)$) {124};	
		\node[ivert] (c101) at ($1.5*(e1)+.5*(e2)+1.5*(e3)$) {134};	
		\node[ivert] (c011) at ($.5*(e1)+1.5*(e2)+1.5*(e3)$) {234};	
		\node[ivert] (c111) at ($1.5*(e1)+1.5*(e2)+1.5*(e3)$) {4};	

		\draw[-, gray]
			(v000) edge (v100) edge (v010) edge (v001)
			(v100) edge (v200) edge (v110) edge (v101)
			(v010) edge (v110) edge (v020) edge (v011)
			(v001) edge (v101) edge (v011) edge (v002)
			(v200) edge (v210) edge (v201)
			(v020) edge (v120) edge (v021) 
			(v002) edge (v102) edge (v012) 
			(v110) edge (v111) edge (v210) edge (v120)
			(v011) edge (v111) edge (v021) edge (v012)
			(v210) edge (v220) edge (v211)
			(v120) edge (v220) edge (v121)
			(v021) edge (v022) edge (v121)
			(v012) edge (v022) edge (v112)
			(v222) edge (v221) edge (v212) edge (v122)
			(v221) edge (v220) edge (v211) edge (v121)
			(v122) edge (v121) edge (v112) edge (v022)
		;
		\draw[-, line width=1pt]
			(c000) edge[ black] (v000) edge[ blue] (v110) edge[ green] (v101) edge[ red] (v011) 
			(c110) edge[ black] (v110) edge[ blue] (v220) edge[ green] (v211) edge[ red] (v121) 
			(c101) edge[ black] (v101) edge[ blue] (v211) edge[ green] (v202) edge[ red] (v112) 
			(c011) edge[ black] (v011) edge[ blue] (v121) edge[ green] (v112) edge[ red] (v022) 
			(c111) edge[ black] (v222) edge[ red] (v211) edge[ green] (v121) edge[ blue] (v112) 
			(c100) edge[ black] (v211) edge[ red] (v200) edge[ green] (v110) edge[ blue] (v101) 
			(c010) edge[ black] (v121) edge[ red] (v110) edge[ green] (v020) edge[ blue] (v011) 
			(c001) edge[ black] (v112) edge[ red] (v101) edge[ green] (v011) edge[ blue] (v002) 
		;
		\draw[-,white, line width=1.8pt]
			(v011) edge (v111)
			(v110) edge (v111)
			(v111) edge (v211) edge (v121) edge (v112)
		;
		\draw[-,gray]
			(v110) edge (v111)
			(v011) edge (v111)
			(v111) edge (v211) edge (v121) edge (v112)
		;
		\draw[-,white, line width=2.5pt]
			(c001) edge[] (v112) 
			(c100) edge[ ] (v211)
			(c101) edge[ ] (v112) 
			(c110) edge[ ] (v110)
			(c111) edge[ ] (v211) edge[ ] (v112) 
			(c011) edge[ ] (v011)
			(c010) edge[ ] (v110) edge[ ] (v011) 
		;
		\draw[-, line width=1pt]
			(c001) edge[ black] (v112) 
			(c100) edge[ black] (v211)
			(c101) edge[ red] (v112) 
			(c110) edge[ black] (v110)		
			(c111) edge[ red] (v211) edge[ blue] (v112) 
			(c011) edge[ black] (v011)
			(c010) edge[ red] (v110) edge[ blue] (v011) 
 		;
		\draw[-,white, line width=1.8pt]
			(v100) edge (v101)
			(v001) edge (v101) 
			(v101) edge (v111) edge (v201) edge (v102)
			(v200) edge (v201)
			(v002) edge (v102)
			(v201) edge (v202) edge (v211)
			(v210) edge (v211)
			(v102) edge (v202) edge (v112)
			(v012) edge (v112)
			(v212) edge (v211) edge (v202) edge (v112)
		;
		\draw[-,gray]
			(v100) edge (v101)
			(v001) edge (v101) 
			(v101) edge (v111) edge (v201) edge (v102)
			(v200) edge (v201)
			(v002) edge (v102)
			(v201) edge (v202) edge (v211)
			(v210) edge (v211)
			(v102) edge (v202) edge (v112)
			(v012) edge (v112)
			(v212) edge (v211) edge (v202) edge (v112)
		;		
		\draw[-,white, line width=2.5pt]
			(c100)  edge[ ] (v200) edge[ ] (v101) 
			(c000) edge[ ] (v110)
			(c101) edge[ ] (v101) edge[ ] (v211) edge[ ] (v202) 
			(c001) edge[ ] (v101)
			(c000) edge[ ] (v000)
			(c001) edge[ ] (v002) 
		;
		\draw[-,black, line width=1pt]
			(c100)  edge[ red] (v200) edge[ blue] (v101) 
			(c000) edge[ blue] (v110)
			(c101) edge[ black] (v101) edge[ blue] (v211) edge[ green] (v202) 
			(c001) edge[ red] (v101)
			(c000) edge[ black] (v000)
			(c001) edge[ blue] (v002) 
		;

		\node[ivert] (v000) at ($0*(e1)+0*(e2)+0*(e3)$) {123\={4}};
		\coordinate (v100) at ($1*(e1)+0*(e2)+0*(e3)$) {};
		\coordinate (v010) at ($0*(e1)+1*(e2)+0*(e3)$) {};
		\coordinate (v001) at ($0*(e1)+0*(e2)+1*(e3)$) {};
		\node[ivert] (v110) at ($1*(e1)+1*(e2)+0*(e3)$) {12};
		\node[ivert] (v011) at ($0*(e1)+1*(e2)+1*(e3)$) {23};
		\node[ivert] (v101) at ($1*(e1)+0*(e2)+1*(e3)$) {13};
		\node[ivert] (v200) at ($2*(e1)+0*(e2)+0*(e3)$) {11};
		\node[ivert] (v020) at ($0*(e1)+2*(e2)+0*(e3)$) {22};
		\node[ivert] (v002) at ($0*(e1)+0*(e2)+2*(e3)$) {33};		
		\coordinate (v210) at ($2*(e1)+1*(e2)+0*(e3)$) {};
		\coordinate (v120) at ($1*(e1)+2*(e2)+0*(e3)$) {};
		\coordinate (v102) at ($1*(e1)+0*(e2)+2*(e3)$) {};		
		\coordinate (v201) at ($2*(e1)+0*(e2)+1*(e3)$) {};
		\coordinate (v021) at ($0*(e1)+2*(e2)+1*(e3)$) {};
		\coordinate (v012) at ($0*(e1)+1*(e2)+2*(e3)$) {};		
		\coordinate (v111) at ($1*(e1)+1*(e2)+1*(e3)$) {};		
		\node[ivert] (v211) at ($2*(e1)+1*(e2)+1*(e3)$) {14};		
		\node[ivert] (v121) at ($1*(e1)+2*(e2)+1*(e3)$) {24};		
		\node[ivert] (v112) at ($1*(e1)+1*(e2)+2*(e3)$) {34};		
		\node[ivert] (v220) at ($2*(e1)+2*(e2)+0*(e3)$) {12\={3}4};
		\node[ivert] (v022) at ($0*(e1)+2*(e2)+2*(e3)$) {\={1}234};
		\node[ivert, fill=white] (v202) at ($2*(e1)+0*(e2)+2*(e3)$) {1\={2}34};
		\coordinate (v221) at ($2*(e1)+2*(e2)+1*(e3)$) {};		
		\coordinate (v212) at ($2*(e1)+1*(e2)+2*(e3)$) {};		
		\coordinate (v122) at ($1*(e1)+2*(e2)+2*(e3)$) {};		
		\node[ivert] (v222) at ($2*(e1)+2*(e2)+2*(e3)$) {44};		
		
		\node[ivert] (c000) at ($.5*(e1)+.5*(e2)+.5*(e3)$) {123};	
		\node[ivert] (c100) at ($1.5*(e1)+.5*(e2)+.5*(e3)$) {1};	
		\node[ivert] (c010) at ($.5*(e1)+1.5*(e2)+.5*(e3)$) {2};	
		\node[ivert] (c001) at ($.5*(e1)+.5*(e2)+1.5*(e3)$) {3};	
		\node[ivert] (c110) at ($1.5*(e1)+1.5*(e2)+.5*(e3)$) {124};	
		\node[ivert] (c101) at ($1.5*(e1)+.5*(e2)+1.5*(e3)$) {134};	
		\node[ivert] (c011) at ($.5*(e1)+1.5*(e2)+1.5*(e3)$) {234};	
		\node[ivert] (c111) at ($1.5*(e1)+1.5*(e2)+1.5*(e3)$) {4};	
		
	\end{tikzpicture}

	\caption{Left: identification of $\lat$ with $A_3$. Right: identification of $\lat \cup \tetb \cup \tetw$ with $\hat A_4 \subset A_4$ (red, green, blue, black correspond to the 1-, 2-, 3-, 4-direction respectively). To improve the readability of the figure, we use shift notation, every label $I$ corresponds to $\shi{I}(z)$ for a point $z \in \Z^5_{-2}$, and we suppressed the 5th coordinate.}
	\label{fig:afour}
\end{figure}

Set $\phi(0)$ = 0. Let $z \in \lat = \Z^3_+$ and $z' \in \tet = \Z^3 + \frac12(1,1,1)$ such that $z$ is a vertex of the tetrahedron $z'$, then we require that for every $i\in [4]$ holds
\begin{align}
	\phi(z) - \phi(z') =
	\begin{cases}
		+a_i & \mbox{if } z-z' = +b_i, \\
		-a_i & \mbox{if } z-z' = -b_i.
	\end{cases}
\end{align}
This defines $\phi$ on all of $\hlat$. Since we gave a definition of $\phi$ via differences, we have to check that this is well-defined. Each closing conditions corresponds to a (planar) rhombus in $\lat \cup \tet$, involving two vertices of $\lat$, a tetrahedron of $\tetb$ and a tetrahedron of $\tetw$. Opposite edges in the rhombus are mapped to the same difference by $\phi$, therefore $\phi$ is well-defined.

Since all the $a_i$ vectors are in $A_5$, so is all of $\phi(\hlat)$. The restriction of $\phi$ to $\lat$ has constant coordinate $z_5 = 0$, the restriction to $\tetb$ has constant coordinate $z_5 = 1$ and the restriction to $\tetw$ has constant coordinate $z_5 = -1$. This reflects the fact that $\lat \simeq A_3$, $\tetb \simeq A_3$ and $\tetw \simeq A_3$. 

Using the identification of $\hlat$ with $\hat A_4$ via $\phi$, we may consider a combined map $(\ub, t, \uw): \hlat \rightarrow \C$.
Fix $m \in [5]$, set 
\begin{align}
	I = [5] \setminus \{m\},
\end{align}
and let $z \in \Z^5_{-2}$. Then the six points
\begin{align}
	\{\shi{i_1,i_2} z \ | \ i_1, i_2 \in I, \ i_1 \neq i_2 \},
\end{align}
are the vertices of an octahedron in $A_4$, and every octahedron may be represented in this way. If $m = 5$, we recover the octahedra that live in only one of the lattices $\lat,\tetb$ or $\tetw$. For $m \neq 5$, the octahedra involve three vertices of $\lat$ and three vertices of either $\tetb$ or $\tetw$.

\begin{definition}\label{def:higherdskp}
	A \emph{dSKP map} $a: A_n \rightarrow \hC$ is a map that satisfies the \emph{dSKP lattice equation}, that is a map such that 
	\begin{align}
	\mr(a(\shi {i_1,i_2} z), a(\shi {i_2,i_3} z), a(\shi {i_1,i_3} z), a(\shi {i_3,i_4} z), a(\shi {i_1,i_4}  z), a(\shi {i_2,i_4} z)) = -1,
	\end{align}
	for all $\{i_1,i_2,i_3,i_4\} \in \binom{[n]}{4}$, and $z \in \Z^{n+1}_{-2}$.
\end{definition}

To show how the $t$-, $\ub$- and $\uw$-maps are a dSKP map on $A_4$, we also need a geometric lemma. From the viewpoint of hyperbolic geometry, this lemma is actually about two related ideal tetrahedra, but we give an elementary proof.

\begin{figure}
	\centering
	\begin{tikzpicture}[scale=2,font=\sffamily\scriptsize]
	
	\coordinate (e1) at (1,0);
	\coordinate (e2) at (30:0.8);
	\coordinate (e3) at (0,1);
	
	\node[ivert] (v000) at ($0*(e1)+0*(e2)+0*(e3)$) {123\={4}};
	\coordinate (v100) at ($1*(e1)+0*(e2)+0*(e3)$) {};
	\coordinate (v010) at ($0*(e1)+1*(e2)+0*(e3)$) {};
	\coordinate (v001) at ($0*(e1)+0*(e2)+1*(e3)$) {};
	\node[ivert] (v110) at ($1*(e1)+1*(e2)+0*(e3)$) {12};
	\node[ivert] (v011) at ($0*(e1)+1*(e2)+1*(e3)$) {23};
	\node[ivert] (v101) at ($1*(e1)+0*(e2)+1*(e3)$) {13};
	\node[ivert] (v020) at ($0*(e1)+2*(e2)+0*(e3)$) {22};
	\coordinate (v210) at ($2*(e1)+1*(e2)+0*(e3)$) {};
	\coordinate (v120) at ($1*(e1)+2*(e2)+0*(e3)$) {};
	\coordinate (v201) at ($2*(e1)+0*(e2)+1*(e3)$) {};
	\coordinate (v021) at ($0*(e1)+2*(e2)+1*(e3)$) {};
	\coordinate (v012) at ($0*(e1)+1*(e2)+2*(e3)$) {};		
	\coordinate (v111) at ($1*(e1)+1*(e2)+1*(e3)$) {};		
	\node[ivert] (v211) at ($2*(e1)+1*(e2)+1*(e3)$) {14};		
	\node[ivert] (v121) at ($1*(e1)+2*(e2)+1*(e3)$) {24};		
	\node[ivert] (v112) at ($1*(e1)+1*(e2)+2*(e3)$) {34};		
	\node[ivert] (v220) at ($2*(e1)+2*(e2)+0*(e3)$) {12\={3}4};
	\node[ivert] (v022) at ($0*(e1)+2*(e2)+2*(e3)$) {\={1}234};
	\coordinate (v221) at ($2*(e1)+2*(e2)+1*(e3)$) {};		
	\coordinate (v122) at ($1*(e1)+2*(e2)+2*(e3)$) {};		
	
	\node[ivert] (c000) at ($.5*(e1)+.5*(e2)+.5*(e3)$) {123};	
	\node[ivert] (c010) at ($.5*(e1)+1.5*(e2)+.5*(e3)$) {2};	
	\node[ivert] (c110) at ($1.5*(e1)+1.5*(e2)+.5*(e3)$) {124};	
	\node[ivert] (c011) at ($.5*(e1)+1.5*(e2)+1.5*(e3)$) {234};	
	
	\draw[-, gray]
	(v000) edge (v100) edge (v010) edge (v001)
	(v100) edge (v110) edge (v101)
	(v010) edge (v110) edge (v020) edge (v011)
	(v001) edge (v101) edge (v011) 
	(v020) edge (v120) edge (v021) 
	(v110) edge (v111) edge (v210) edge (v120)
	(v011) edge (v111) edge (v021) edge (v012)
	(v210) edge (v220) edge (v211)
	(v120) edge (v220) edge (v121)
	(v021) edge (v022) edge (v121)
	(v012) edge (v022) edge (v112)
	(v221) edge (v220) edge (v211) edge (v121)
	(v122) edge (v121) edge (v112) edge (v022)
	;
	\draw[-, line width=1pt]
	(c000)  edge[ blue] (v110) edge[ red] (v011) 
	(c110) edge[ black] (v110) edge[ red] (v121) 
	(c011) edge[ black] (v011) edge[ blue] (v121) 
	;
	\draw[-,white, line width=1.8pt]
	(v011) edge (v111)
	(v110) edge (v111)
	(v111) edge (v211) edge (v121) edge (v112)
	;
	\draw[-,gray]
	(v110) edge (v111)
	(v011) edge (v111)
	(v111) edge (v211) edge (v121) edge (v112)
	;
	\draw[-,white, line width=2.5pt]
	(c110) edge[ ] (v110)
	(c011) edge[ ] (v011)
	;
	\draw[-, line width=1pt]
	(c110) edge[ black] (v110)		
	(c011) edge[ black] (v011)
	;
	\draw[-,white, line width=1.8pt]
	(v100) edge (v101)
	(v001) edge (v101) 
	(v101) edge (v111) 
	(v210) edge (v211)
	(v012) edge (v112)
	;
	\draw[-,gray]
	(v100) edge (v101)
	(v001) edge (v101) 
	(v101) edge (v111) 
	(v210) edge (v211)
	(v012) edge (v112)
	;	
	\draw[-,white, line width=2.5pt]
	(c000) edge[ ] (v110)
	;
	\draw[-,black, line width=1pt]
	(c000) edge[ blue] (v110)
	;

	\node[ivert] (v000) at ($0*(e1)+0*(e2)+0*(e3)$) {123\={4}};
	\coordinate (v100) at ($1*(e1)+0*(e2)+0*(e3)$) {};
	\coordinate (v010) at ($0*(e1)+1*(e2)+0*(e3)$) {};
	\coordinate (v001) at ($0*(e1)+0*(e2)+1*(e3)$) {};
	\node[ivert] (v110) at ($1*(e1)+1*(e2)+0*(e3)$) {12};
	\node[ivert] (v011) at ($0*(e1)+1*(e2)+1*(e3)$) {23};
	\node[ivert] (v101) at ($1*(e1)+0*(e2)+1*(e3)$) {13};
	\node[ivert] (v020) at ($0*(e1)+2*(e2)+0*(e3)$) {22};
	\coordinate (v210) at ($2*(e1)+1*(e2)+0*(e3)$) {};
	\coordinate (v120) at ($1*(e1)+2*(e2)+0*(e3)$) {};
	\coordinate (v201) at ($2*(e1)+0*(e2)+1*(e3)$) {};
	\coordinate (v021) at ($0*(e1)+2*(e2)+1*(e3)$) {};
	\coordinate (v012) at ($0*(e1)+1*(e2)+2*(e3)$) {};		
	\coordinate (v111) at ($1*(e1)+1*(e2)+1*(e3)$) {};		
	\node[ivert] (v211) at ($2*(e1)+1*(e2)+1*(e3)$) {14};		
	\node[ivert] (v121) at ($1*(e1)+2*(e2)+1*(e3)$) {24};		
	\node[ivert] (v112) at ($1*(e1)+1*(e2)+2*(e3)$) {34};		
	\node[ivert] (v220) at ($2*(e1)+2*(e2)+0*(e3)$) {12\={3}4};
	\node[ivert] (v022) at ($0*(e1)+2*(e2)+2*(e3)$) {\={1}234};
	\coordinate (v221) at ($2*(e1)+2*(e2)+1*(e3)$) {};		
	\coordinate (v122) at ($1*(e1)+2*(e2)+2*(e3)$) {};		
	
	\node[ivert] (c000) at ($.5*(e1)+.5*(e2)+.5*(e3)$) {123};	
	\node[ivert] (c010) at ($.5*(e1)+1.5*(e2)+.5*(e3)$) {2};	
	\node[ivert] (c110) at ($1.5*(e1)+1.5*(e2)+.5*(e3)$) {124};	
	\node[ivert] (c011) at ($.5*(e1)+1.5*(e2)+1.5*(e3)$) {234};		
	
	\end{tikzpicture}
	\hspace{4mm}
	\includegraphics[scale=0.3]{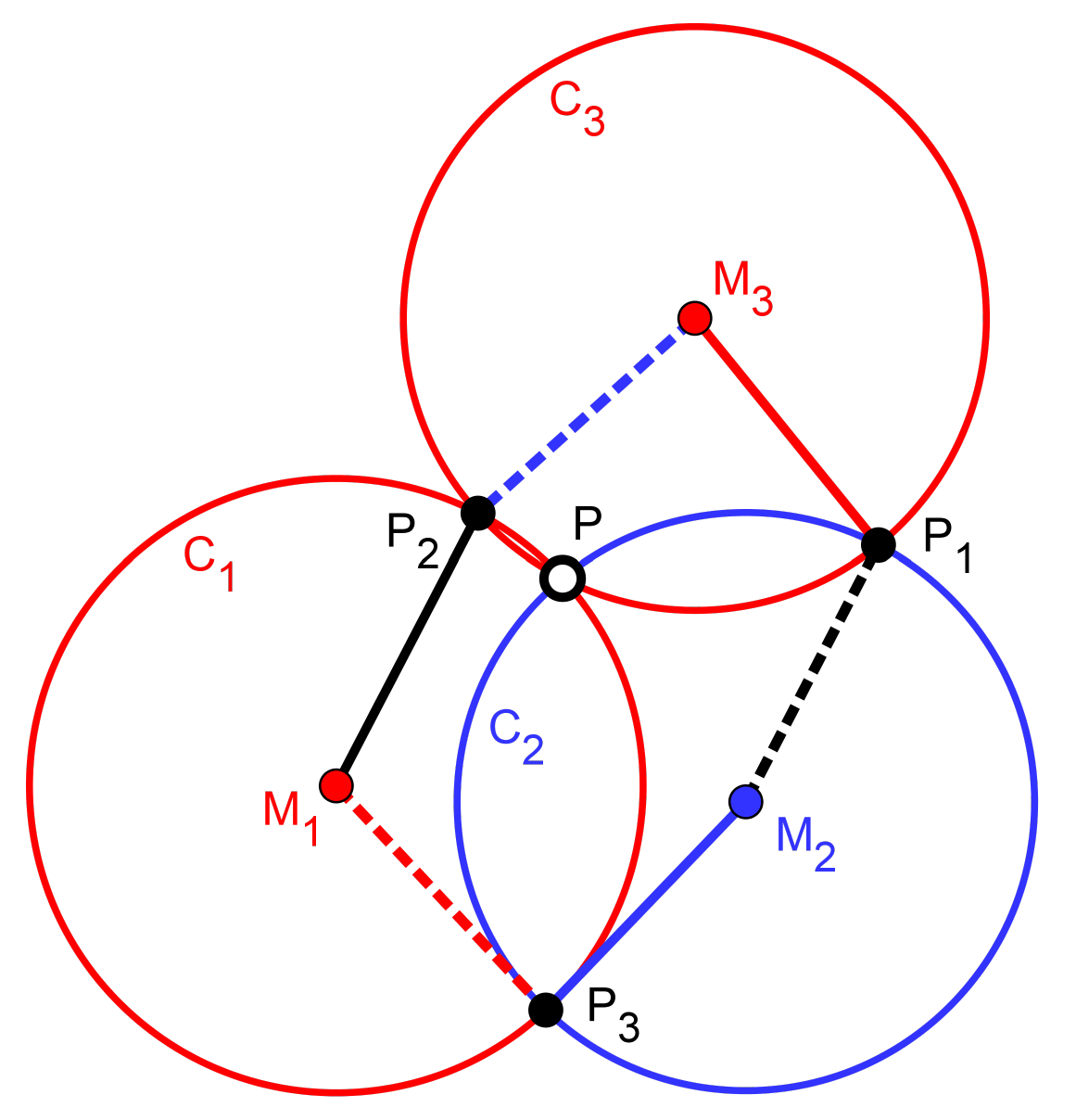}
	
	\caption{Left: an octahedron in $\hat A_4$, with three vertices corresponding to $\lat$ and three to $\tetb$. Right: the corresponding geometric configuration with labels as in Lemma~\ref{lem:tumr}.}
	\label{fig:tulemma}
\end{figure}

\begin{lemma}\label{lem:tumr}
	Consider three circles $C_1,C_2,C_3$ that intersect in a point $P$, and let $P_{1}, P_{2}, P_{3}$ be the other three intersection points, see Figure~\ref{fig:tulemma}. Let $M_1,M_2,M_3$ be the corresponding circle centers. Then we have
	\begin{align}
	\mr(P_1,M_3,P_2,M_1,P_3,M_2) = -1.
	\end{align}
\end{lemma}
\proof{
	Since multi-ratios are invariant under Möbius transformations, we may normalize the configuration with a Möbius transformation $M$ such that $M(P) = \infty$ and $M(\infty) = 0$. As a consequence, we may assume $C_1,C_2,C_3$ are straight lines and $M_1, M_2, M_3$ are the reflections of $0$ about the lines $P_2P_3$, $P_1P_3$, $P_1P_2$. In complex numbers, this implies that
	\begin{align}
	M_3 = P_1 - \bar P_1 \frac{P_2- P_1}{\bar P_2 - \bar P_1},
	\end{align}
	and analogous formulas for $M_1$ and $M_2$. Moreover, we calculate that
	\begin{align}
	M_3 - P_1 = - \bar P_1 \frac{P_2- P_1}{\bar P_2 - \bar P_1}, \quad M_3 - P_2 = - \bar P_2 \frac{P_1- P_2}{\bar P_1 - \bar P_2}.
	\end{align}
	In turn, this implies
	\begin{align}
	\frac{M_3 - P_1}{M_3 - P_2} = \frac{\bar P_1}{\bar P_2}.
	\end{align}
	Inserting this and the two other analogous ratios into the multi-ratio in the claim yields the desired result. \qed
}

\begin{remark}
	Note that by combining several copies of the equation of Lemma~\ref{lem:tumr} in one cube of a Miquel six circle configuration, one may obtain alternative proofs for both Theorem~\ref{th:tdskp} and Theorem~\ref{th:ytotoda}.
\end{remark}

\begin{theorem}\label{th:backlund}
	Given a Miquel map, the associated dSKP maps $\ub, t, \uw$ are a dSKP map on $\hat A_4$, via the identification $\phi$ of $\hat A_4$ with $\hlat$.
\end{theorem}
\proof{
	For octahedra that do not involve direction $5$, the dSKP equation follows from Theorem~\ref{th:tdskp} and Theorem~\ref{th:udskp}. Each other octahedron involves three circles intersecting in a common point. More specifically, consider $z \in \Z^5_{-2}$ and let $J \in \binom{[5]}{4}$ be such that $J \cup \{5\} = \{i_1,i_2,i_3,5\}$. Then the three vertices $\shi{i,5}z$ for $i \in J$ and the three vertices $\shi{i,j}z$ for $i,j \in J$, $i\neq j$ make up such an octahedron. In Figure~\ref{fig:tulemma}, we illustrate an example of such an octahedron and how this octahedron is situated in $\lat \cup \tetb$. Therefore, the dSKP equation involves three circle centers of $t$ and three intersection points of either $\ub$ or $\uw$ as required for Lemma~\ref{lem:tumr} to apply, which concludes the proof.\qed
}

\begin{remark}
	In a sense, Theorem~\ref{th:backlund} states that $\ub$ (and $\uw$) is a \emph{Bäcklund transform} of $t$ in $A_4$, since these are defined on adjacent copies of $A_3$, similar to how Bäcklund transforms are usually two solutions on adjacent copies of $\Z^2$ in $\Z^3$.
	Moreover, using Lemma~\ref{lem:tumr} one can show that $t$ together with 1-dimensional initial data of $\ub$, determines all of $\ub$, and vice versa. The same holds with $\uw$. In Section~\ref{sec:toda}, we discuss new and old Y-systems for Miquel maps. It is currently unclear if there is a functional relation between these Y-systems. However, since the dSKP maps are related, we suspect the same holds for the Y-systems. We think the results of this section may support research in that direction.
\end{remark}

\section{Y-systems} \label{sec:toda}

Y-systems were introduced in \cite{zamolodchikov}, see also further references in \cite{abs} or \cite{knsysystems} for an overview. 

\begin{definition}\label{def:toda}
	We say a map $U: \Z^3_\pm \rightarrow \C$ is a \emph{Y-system} if it satisfies
	\begin{align}
		U(\shi 3 z)  U(\shi {\sm 3} z) = \frac{(1+U(\shi 2 z))(1 + U(\shi {\sm 2} z))}{(1 + U^{-1}(\shi 1 z))(1 + U^{-1}(\shi {\sm 1} z))}, \label{eq:ysystems}
	\end{align}
	for all $z \in \Z_{\mp}$.
\end{definition}

Let us recall the literature result relating Y-systems and Miquel maps.

\begin{definition}\label{def:yvariables}
	Consider a Miquel map and let $t$ be the associated $t$-map. The \emph{$Y$-variables} $Y: \oct \rightarrow \R$, are defined by
	\begin{align}
		Y(z) \coloneqq - \frac{(t(\shi{1,3} z) - t(z)))(t(\shi{\sm 1, 3} z) - t(z)))}{(t(\shi{2,3} z) - t(z)))(t(\shi{\sm 2, 3} z) - t(z)))}.\label{eq:yvariable}
	\end{align}
\end{definition}

It was shown in \cite{amiquel,klrr} that the $Y$-variables are real-valued. This is a condition on the relative angles of the four lines through $t(z)$ and through $t(\shi{1,3} z)$, $t(\shi{2,3} z)$, $t(\shi{\sm 1, 3} z)$, $t(\shi{\sm 2, 3} z)$ respectively. Since the intersection points $ p(\shi{\sm 1, \sm 2}z)$, $ p(\shi{\sm 2}z)$, $p(z)$, $ p(\shi{\sm 1}z)$ are obtained from consecutively reflecting about the four lines, the composition of these four reflections is the identity. This imposes an angle condition on the four angles of consecutive lines, which in turn implies that $Y(z)$ is real-valued.

The first relation between Miquel maps and Y-systems was found in \cite{amiquel} and independently in \cite{klrr}. We recall this relation in the next theorem.

\begin{theorem}
	The $Y$-variables of a Miquel map are a Y-system.
\end{theorem}

\begin{figure}
	\centering
	\includegraphics[scale=0.4]{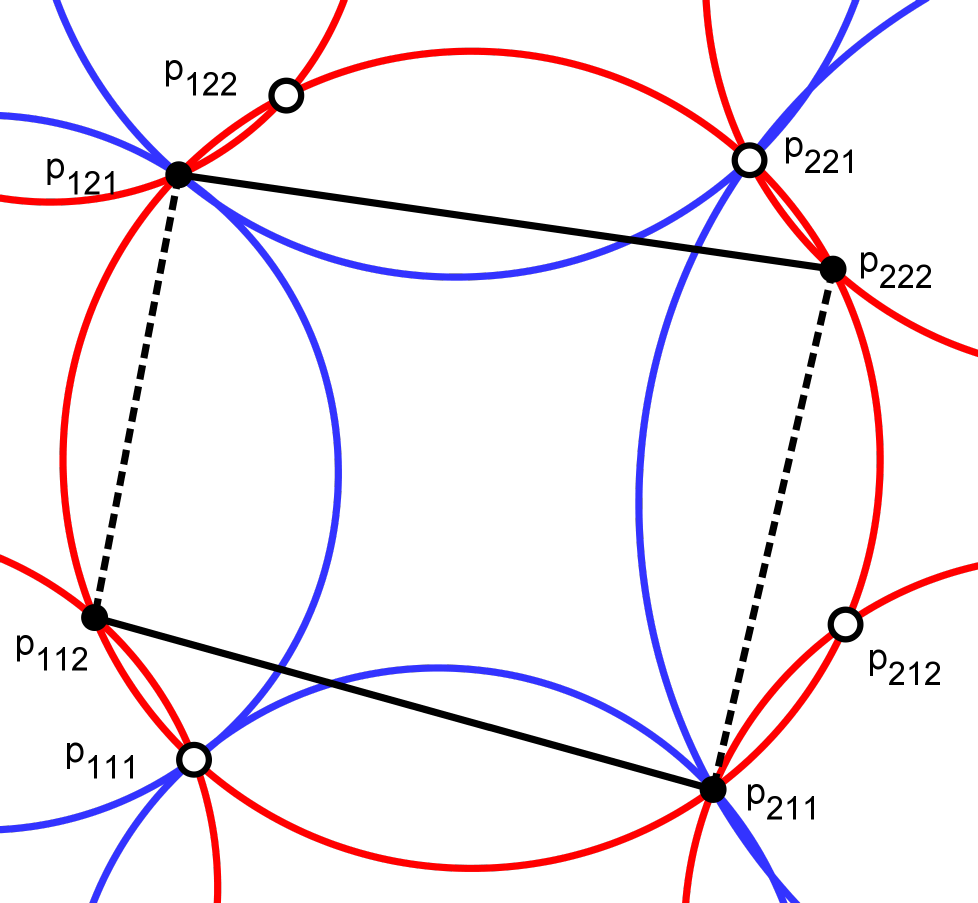}
	\hspace{5mm}
	\includegraphics[scale=0.4]{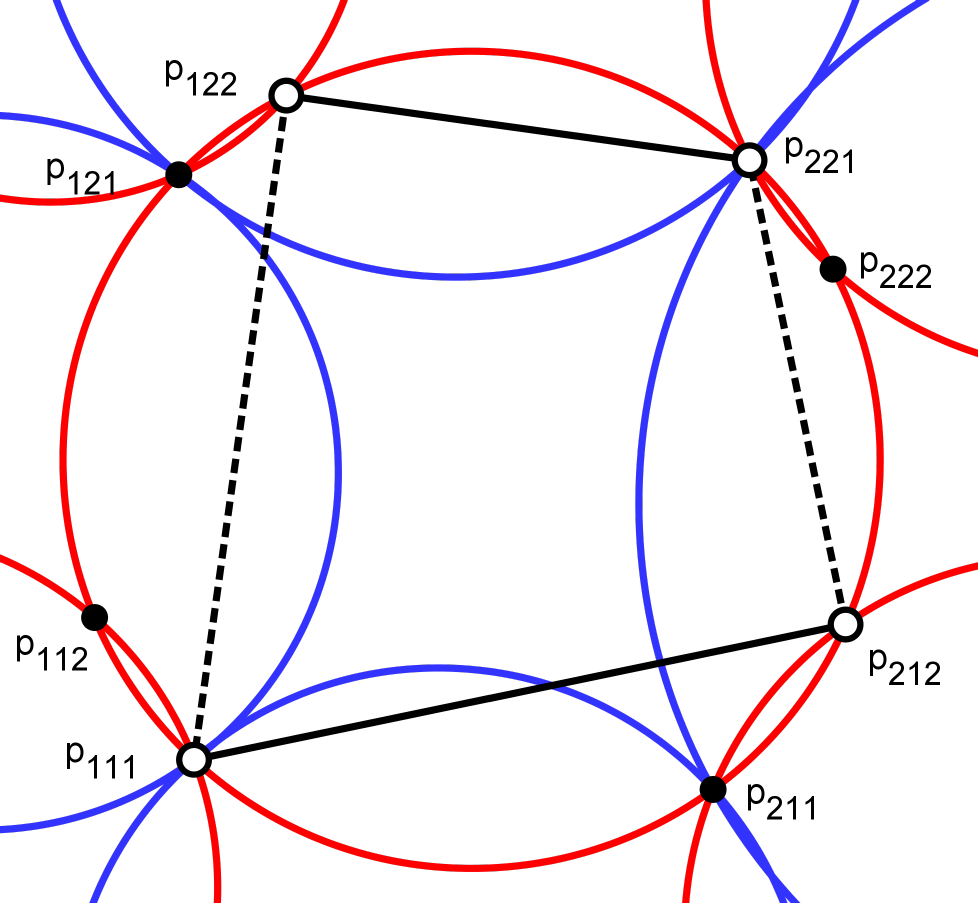}
	\caption{The points that define $\Xb({2,2,2})$ (left) and $\Xw({2,2,2})$ (right) on the circle $c_{2,2,2}$.}
	\label{fig:ax}
\end{figure}

Now, let us present an analogous result for the $\ub$- and $\uw$-maps.

\begin{definition} \label{def:xvariables}
	Consider a Miquel map and the associated $\ub, \uw$ maps.
	The \emph{$\Xb$-variables} $\Xb: \lat \rightarrow \R$, are defined by
	\begin{align}
		\Xb(z) \coloneqq -\cro( \ub(\shi{\sm 1,\sm 2} z),  \ub(\shi{\sm 2,\sm 3}z), \ub(z),  \ub(\shi{\sm 1,\sm 3}z) ). \label{eq:xbvars}
	\end{align}
	Analogously, the \emph{$\Xw$-variables} $\Xw: \lat \rightarrow \R$, are defined by
	\begin{align}
		\Xw(z) \coloneqq -\cro(\uw(\shi{\sm 1,\sm 2,\sm 3}z),  \uw(\shi{\sm 2}z),  \uw(\shi{\sm 3}z),  \uw(\shi{\sm 1}z)).
	\end{align}
\end{definition}

As mentioned in the introduction, the $\Xb$- and $\Xw$-variables are real-valued, since they are cross-ratios of four points on a circle, see also Figure~\ref{fig:ax}.

\begin{theorem}\label{th:xtoda}
	The $\Xb$-variables and the $\Xw$-variables of a Miquel map are a Y-system.
\end{theorem}

\begin{proof}
	In Theorem~\ref{th:udskp} we have shown that $\ub$-maps are dSKP maps. The $\Xb$-variables that we defined in Definition~\ref{def:xvariables} for $\ub$-maps coincide with variables introduced in \cite[Section~7]{abs}, where they are denoted by $h$ (albeit with a minus sign). Moreover, also in \cite{abs}  it was shown that these $h$-variables, and hence our $\Xb$-variables form a Y-system. The proof is based on a direct calculation, showing that Equation~\eqref{eq:ysystems} is equivalent to four instances of Equation~\eqref{eq:udskp}. The proof for the $\Xw$-variables proceeds analogously, since they also coincide with the $h$-variables. 
\end{proof}

Since the $X$-variables constitute a Y-system, we see that positivity of the $X$-variables is also preserved under Miquel dynamics. Of course, whether an $X$-variable is positive or not depends on the order in which the four points of $\ub$ (resp.\ $\uw$) appear on the corresponding circle. For example, in Figure~\ref{fig:miquelcps}, the $\Xb$-variables are all positive since the points of $\ub$ appear in cyclic order on each circle, and the same holds for the $\Xw$-variables. However, it is not directly clear in what manner this characterizes the corresponding circle pattern. We investigate positivity conditions in upcoming work \cite{ammiquel}.

Since $\Xb$ and $\Xw$ satisfy the same lattice equation (as a Y-system), it is possible that $\Xb(z) = \Xw(z)$ for all $z \in \lat$. A practical technique to construct ``regular looking'' larger patches of a circle pattern is to just take a rectangular lattice for the intersection points of the circle pattern. Experiments show that Miquel maps that contain such a rectangular circle pattern do indeed satisfy $\Xb = \Xw$, which leads us to ask a question.

\begin{question}
	Is there a geometric characterization of all circle patterns or Miquel maps that satisfy $\Xb = \Xw$?
\end{question}
This question may relate to our results on $h$-Miquel maps, see Section~\ref{sec:harmonic}, we did not investigate this direction further.

Next, we explain how the $Y$- and $X$-variables determine a Miquel map up to boundary data.
The \emph{$k$-th layer} of a Y-system $U: \Z^3_\pm \rightarrow \C$ is the restriction of $U$ to $\{z\in \Z^3_\pm \ | \ z_3 \in \{k,k+1\}\}$.

Let us begin with the $Y$-variables. Let $t_k$ be the restriction of a $t$-map $t$ to all $z\in \lat$ such that $z_3 \in \{k,k+1\}$. Then $t_k$ determines all of $t$, since Theorem~\ref{th:tdskp} determines $t_{k+1}$ and $t_{k-1}$ given $t_k$. The conical net $t$ does not determine the intersection points $c$. Instead, we may choose one of the intersection points freely and only then does the $t$-map determine all of the Miquel map.

\begin{theorem}\label{th:ytotoda}
	Let $(c,p)$ be a Miquel map, and let $t_k$ be the conical net of $(c_k,p_k)$. 
	The $t$-map is uniquely determined by the $0$-th layer of the $Y$-variables and the boundary data
	\begin{align}
		\{  t(z)  \ | \ z \in \lat, \ z_2 = 0,1, \ z_3 = 0,1 \}.
	\end{align}
	Generically, any real-valued Y-system corresponds to the $Y$-variables of a $t$-map.
\end{theorem}

\begin{proof}
	Let us explain how the $0$-th layer of the $Y$-variables determines $t_0$ up to boundary data (see also Figure~\ref{fig:boundary}). Assume we know $Y(z)$ with $z_3 = 0$ and $t$ on $\shi{1,3}z, \shi{\sm 1,3}z, \shi{\sm 2,3}z$. Then we may solve Equation~\eqref{eq:yvariable} uniquely for $\shi{2,3}z$. For $z_3 = 1$, a quick calculation shows that
	\begin{align}
		Y(z) = - \frac{(t(\shi{2,\sm 3} z) - t(z))(t(\shi{\sm 2,\sm 3} z) - t(z))}{(t(\shi{1,\sm 3} z) - t(z))(t(\shi{\sm 1,\sm 3} z) - t(z))}. \label{eq:yy}
	\end{align}
	Therefore, if we know $Y(z)$ with $z_3 = 1$ the point $\shi{2,\sm 3} z$ is determined by $\shi{1,\sm 3}z, \shi{\sm 1,\sm 3}z, \shi{\sm 2,\sm 3}z$. Now, if we choose boundary data as in the claim, all of $t_0$ is determined by the $0$-th layer of the $Y$-variables. Since the $Y$-variables are real, $t_0$ is indeed a conical net. The whole $t$-map is then determined by Miquel dynamics.
\end{proof}

Now let us prepare the analogous result for the $\Xb$-variables (and analogously the $\Xw$-variables). Assume we know $\ub$ but not $\uw$. For every $z\in \lat$, the circle $c(z)$ contains four points of $\ub$, therefore all the circles are determined by $\ub$. As a consequence, also all the intersection points $p$ are determined. Thus, $\ub$ (or $\uw$) determine the Miquel map uniquely.

Recall that $p_k(z)$ is the restriction of $p$ to all $z\in \lat$ with $z_3 = k$. We define the restriction  $p^\bullet_k$ slightly differently, as the restriction of $\ub$ to $z\in \lat$ with $z_3 \in \{k,k+1\}$. With this definition, the restriction $p^\bullet_k$ determines all of $\ub$, since we may use Theorem~\ref{th:udskp} to determine $p^\bullet_{k+1}$ and $p^\bullet_{k-1}$ from $p^\bullet_k$.

\begin{figure}
	\centering
	\begin{tikzpicture}[scale=1]
		\draw[step=1.0,black] (-3,-3) grid (3,3);
		\node[bvert, fill=gray] at (-3,0) {};
		\node[bvert, fill=gray] at (-2,0) {};
		\node[bvert, fill=gray] at (-1,0) {};
		\node[bvert, fill=gray] at (0,0) {};
		\node[bvert, fill=gray] at (1,0) {};
		\node[bvert, fill=gray] at (2,0) {};
		\node[bvert, fill=gray] at (3,0) {};
		\node[bvert, fill=gray] at (-3,1) {};
		\node[bvert, fill=gray] at (-2,1) {};
		\node[bvert, fill=gray] at (-1,1) {};
		\node[bvert, fill=gray] at (0,1) {};
		\node[bvert, fill=gray] at (1,1) {};
		\node[bvert, fill=gray] at (2,1) {};
		\node[bvert, fill=gray] at (3,1) {};
		\draw[blue, -, line width=1.5pt, line cap=round] 
			(1,1) -- (0,1) -- (-1,1)
		;
		\draw[blue, densely dashed, line width=2pt, line cap=round] 
			(0,2) -- (0,1) -- (0,0)
		;		
	\end{tikzpicture}
	\hspace{2cm}
	\begin{tikzpicture}[scale=1]
		\fill[gray] (-3,3) -- (-1,3) -- (-1, 2) -- (1,2) -- (1,3) -- (3,3) -- (3,2) -- (2,2) -- (2,1) -- (1,1) -- (1,0) -- (2,0) -- (2,-1) -- (3,-1) -- (3,-3) -- (2,-3) -- (2,-2) -- (1,-2) -- (1,-1) -- (-1,-1) -- (-1,-2) -- (-2,-2) -- (-2,-3) -- (-3,-3) -- (-3,-1) -- (-2,-1) -- (-2,0) -- (-1,0) -- (-1,1) -- (-2,1) -- (-2,2) -- (-3,2) -- (-3,2);
		\draw[step=1.0,black] (-3,-3) grid (3,3);
		\node[bvert, blue] at (0,0) {};
		\node[bvert, blue] at (2,0) {};
		\node[bvert, blue] at (-2,0) {};
		\node[bvert, blue] at (0,-2) {};
		\node[bvert, blue] at (2,-2) {};
		\node[bvert, blue] at (-2,-2) {};
		\node[bvert, blue] at (0,2) {};
		\node[bvert, blue] at (2,2) {};
		\node[bvert, blue] at (-2,2) {};
		\node[bvert, blue] at (1,1) {};
		\node[bvert, blue] at (3,1) {};
		\node[bvert, blue] at (-1,1) {};
		\node[bvert, blue] at (-3,1) {};
		\node[bvert, blue] at (1,-1) {};
		\node[bvert, blue] at (3,-1) {};
		\node[bvert, blue] at (-1,-1) {};
		\node[bvert, blue] at (-3,-1) {};
		\node[bvert, blue] at (1,3) {};
		\node[bvert, blue] at (3,3) {};
		\node[bvert, blue] at (-1,3) {};
		\node[bvert, blue] at (-3,3) {};
		\node[bvert, blue] at (1,-3) {};
		\node[bvert, blue] at (3,-3) {};
		\node[bvert, blue] at (-1,-3) {};
		\node[bvert, blue] at (-3,-3) {};
		\node[bvert, red] at (0,1) {};
		\node[bvert, red] at (2,1) {};
		\node[bvert, red] at (-2,1) {};
		\node[bvert, red] at (0,-1) {};
		\node[bvert, red] at (2,-1) {};
		\node[bvert, red] at (-2,-1) {};
		\node[bvert, red] at (0,3) {};
		\node[bvert, red] at (2,3) {};
		\node[bvert, red] at (-2,3) {};
		\node[bvert, red] at (0,-3) {};
		\node[bvert, red] at (2,-3) {};
		\node[bvert, red] at (-2,-3) {};
		\node[bvert, red] at (1,0) {};
		\node[bvert, red] at (3,0) {};
		\node[bvert, red] at (-1,0) {};
		\node[bvert, red] at (-3,0) {};
		\node[bvert, red] at (1,2) {};
		\node[bvert, red] at (3,2) {};
		\node[bvert, red] at (-1,2) {};
		\node[bvert, red] at (-3,2) {};
		\node[bvert, red] at (1,-2) {};
		\node[bvert, red] at (3,-2) {};
		\node[bvert, red] at (-1,-2) {};
		\node[bvert, red] at (-3,-2) {};
		
		\draw[red, densely dashed, line width=1.5pt, line cap=round] 
			(-1.5,1.5) -- (-1.5,2.5)
			(-0.5,1.5) -- (-0.5,2.5)
		;
		\draw[red, -, line width=1.5pt, line cap=round] 
			(-1.5,1.5) -- (-0.5,1.5)
			(-1.5,2.5) -- (-0.5,2.5)
		;

		\draw[blue, -, line width=1.5pt, line cap=round] 
			(-0.5,0.5) -- (0.5,1.5)
			(-1.5,-0.5) -- (-0.5,-0.5)
			(-0.5,-1.5) -- (0.5,-0.5)
			(1.5,0.5) -- (0.5,0.5)
		;
		\draw[blue, densely dashed, line width=1.5pt, line cap=round] 
			(-0.5,0.5) -- (-1.5,-0.5)  
			(-0.5,-0.5) -- (-0.5,-1.5)
			(0.5,-0.5) -- (1.5,0.5) 
			(0.5,0.5) -- (0.5,1.5)
		;
	\end{tikzpicture}
	\caption{Left: the gray vertices represent the boundary data chosen in Theorem~\ref{th:ytotoda}. Each vertex comes with an instance of Equation~\eqref{eq:yvariable} or Equation~\eqref{eq:yy} (blue lines). Together the equations allow propagation to the whole lattice. Right: the gray faces contain the boundary data chosen in Theorem~\ref{th:xboundary}. Each red vertex comes with an instance of Equation~\eqref{eq:xbvars} (red lines); each blue vertex with an instance of Equation~\eqref{eq:xmreight} (blue lines), as well as a 90° rotated instance. Together the equations allow propagation to the whole lattice.}
	\label{fig:boundary}
\end{figure}

\begin{theorem} \label{th:xboundary}
	Let $(c,p)$ be a Miquel map. Then $(c,p)$ is uniquely determined by the $0$-th layer of the $\Xb$-variables and the boundary data
	\begin{align}
		\bigl\{  \ub(z) \ | \ z \in \lat, \ z_1+z_2 \in \{0,1\}, \ z_1 - z_2 = \{0,1\}, \ z_3 = \{0,1\} \bigr\}.
	\end{align}
	Generically, any real-valued Y-system corresponds to the $\Xb$-variables of some Miquel map. The analogous two statements hold for $\uw$-maps and $\Xw$-variables.
\end{theorem}
\proof{
	Let us explain how the $0$-th layer of the $\Xb$-variables determines $p^\bullet_0$ up to boundary data.
	The boundary data in the claim corresponds to two pairs of consecutive diagonals of $p^\bullet_0$, as represented by the gray quads in Figure~\ref{fig:boundary}. 	
	Assume we know $\Xb(z)$ for $z\in \Z^3_-$ with $z_3 = 1$ and $\ub$ on $\shi{\sm 1}z$, $\shi{\sm 2} z$, $\shi{\sm1,\sm 2, \sm3} z$. Then we may solve Equation~\eqref{eq:xbvars} for  $\ub(\shi{\sm3}z)$. However, we need a second multi-ratio equation to propagate to all of $p^\bullet_0$.
	Let $z \in \Z^3_-$ and consider the three quantities
	\begin{align}
		\cro( \ub(\shi{3} z), \ub(\shi{1} z), \ub(\shi{1,2,3} z), \ub(\shi{2} z)),\\
		\mr ( \ub(\shi{\sm 1} z), \ub(\shi{\sm 3} z), \ub(\shi{\sm 2} z), \ub(\shi{1} z), \ub(\shi{3} z), \ub(\shi{2} z)),\\
		\mr ( \ub(\shi{1,2,3} z), \ub(\shi{1} z), \ub(\shi{1,1,2} z), \ub(\shi{1,2,\sm 3} z), \ub(\shi{1,2,2} z), \ub(\shi{2} z)).
	\end{align}
	The cross-ratio is equal to $\Xb( \shi{1,2,3} z)$ (Definition~\ref{def:xvariables}), the multi-ratios are both equal to $-1$, since they are instances of the dSKP equation (Theorem~\ref{th:udskp}). Therefore the product of the three quantities is equal to $\Xb( \shi{1,2,3} z)$, but also due to edge-cancellations to
	\begin{align}
		\mr(\ub(\shi{\sm 1} z), \ub(\shi{\sm 3} z), \ub(\shi{\sm 2} z), \ub(\shi{1} z), \ub(\shi{1,1,2} z), \ub(\shi{1,2, \sm 3} z), \ub(\shi{1,2,2} z), \ub(\shi{2} z)). \label{eq:xmreight}
	\end{align}
	Therefore, we know the values of these multi-ratios, which is $\Xb(\shi{1,2,3} z)$. Note that these multi-ratios involve only points of $p^\bullet_0$, which we means we can solve any one of them for one of the points involved, assuming we know the other seven points. It is now straight-forward to verify that we may propagate the initial data to all of $p^\bullet_0$.
	As explained before, $p^\bullet_0$ determines all of $\ub$, which in turn determines $(c,p)$ uniquely. Note that the fact that the $\Xb$-variables are real ensures that the appropriate points of $\ub$ are on circles.\qed
}

\section{Integrable circle patterns} \label{sec:icp}

In Section~\ref{sec:dskp} we showed that the $\ub$-map defined on $\tetb$ is a dSKP map. Then we introduced the $\Xb$-variables in Section~\ref{sec:toda} and showed that the $\Xb$-variables are a Y-system. As we explained, the $\Xb$-variables coincide with the general Y-system associated to a dSKP map in \cite{abs}. Note that if we identify the black tetrahedra $\tetb$ with a copy of $\lat$, let us say $\lat'$, then the $X$-variables are cross-ratios of one type of tetrahedra of $\lat'$. However, there are two types of tetrahedra in $\lat'$. Therefore, we may also consider cross-ratios of the other type of tetrahedra.

\begin{figure}
	\centering
	\includegraphics[scale=0.7]{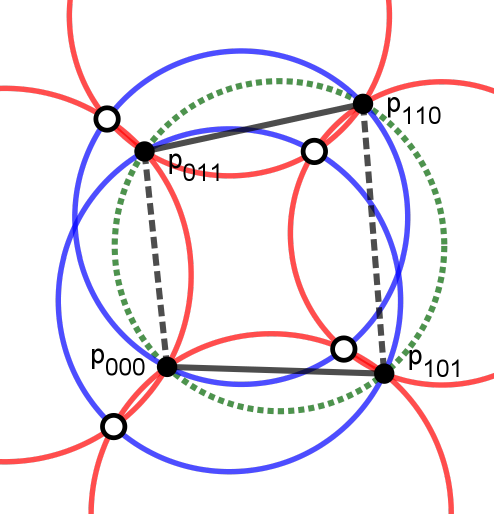}
	\hspace{8mm}
	\includegraphics[scale=0.7]{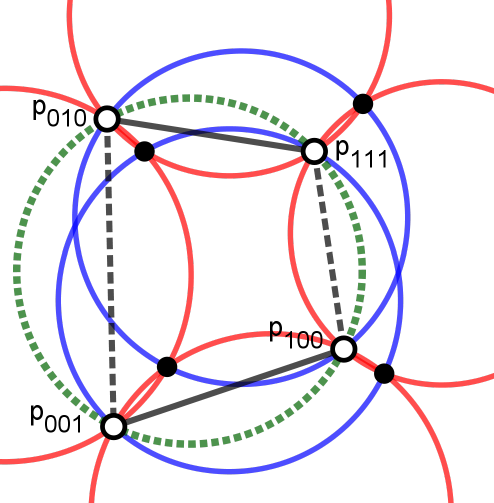}
	\caption{The points that define $\Wb({1,1,1})$ (left) and $\Ww({1,1,1})$ (right) in the Miquel configuration at the octahedron $(1,1,1)$. This is a special case of a Miquel configuration in an integrable circle pattern, in which two additional (green, dotted) circles exist.}
	\label{fig:intcp}
\end{figure}

\begin{definition}\label{def:wvariables}
	Consider a Miquel map and the associated $\ub, \uw$ maps.
	The \emph{$\Wb$-variables} $\Wb: \oct \rightarrow \C$, are defined by
	\begin{align}
		\Wb(z) \coloneqq -\cro(\ub(\shi{\sm 1,\sm 2,\sm 3}z),  \ub(\shi{\sm 2}z),  \ub(\shi{\sm 3}z),  \ub(\shi{\sm 1}z)).
	\end{align}
	Analogously, the \emph{$\Ww$-variables} $\Xw: \oct \rightarrow \C$, are defined by
	\begin{align}
		\Ww(z) \coloneqq -\cro( \uw(\shi{\sm 1,\sm 2}z),  \uw(\shi{\sm 2,\sm 3}z), \uw(z),  \uw(\shi{\sm 1,\sm 3}z) ).
	\end{align}	
\end{definition}

Note that the $W$-variables are generally complex-valued, not real-valued. However, they satisfy the same recursion law as the $X$-variables.

\begin{theorem}\label{th:wtoda}
	The $\Wb$-variables and the $\Ww$-variables of a Miquel map are a Y-system.
\end{theorem}
\proof{
	Same as proof of Theorem~\ref{th:udskp}.\qed
}

The $\Wb$- and $\Ww$-variables are actually directly related.

\begin{theorem}\label{th:wconjugate}
	The $\Wb$- and $\Ww$-variables are complex conjugates, that is for all $z \in \oct$ holds $\Ww(z) = \overline{\Wb}(z)$.
\end{theorem}
\proof{
	To simplify notation, define
	\begin{align}
		P &= \ub(\shi{\sm 1,\sm 2,\sm 3} z), & P_{3,1} &=  \ub(\shi{\sm 2} z), & P_{1,2} &=  \ub(\shi{\sm 3} z), & P_{2,3} &=  \ub(\shi{\sm 1} z),\\
		P_{3} &=  \uw(\shi{\sm 1,\sm 2} z), & P_1 &=  \uw(\shi{\sm 2,\sm 3} z), & P_{1,2,3} &= \uw(z), & P_2 &=  \uw(\shi{\sm 1,\sm 3} z).
	\end{align}		
	We normalize $P = 0$, $P_{1,2,3} = \infty$. Let $i \in [3]$ and consider the crossratios
	\begin{align}
		\lambda_{i,i+1} = \cro(P_i,P,P_{i+1}, P_{i,i+1}) \in \R.
	\end{align}
	Solving this equation for $P_{i,i+1}$ gives
	\begin{align}
		P_{i,i+1} = \frac{\lambda_{i,i+1} - 1}{\lambda_{i,i+1}P_i^{-1} - P_{i+1}^{-1}}.\label{eq:givencr}
	\end{align}
	However, the three cross-ratios $\lambda_{1,2}, \lambda_{2,3}, \lambda_{3,1}$ are not independent, since for each $i\in[3]$ the three points $P_i,P_{i-1,i},P_{i,i+1}$ need to be on a line. The colinearity condition is equivalent to
	\begin{align}
		\frac{P_{i-1,i}-P_i}{P_{i,i+1}-P_i} \in \R.
	\end{align}
	Inserting Expressions~\eqref{eq:givencr} yields
	\begin{align}
		\frac{(\bar P_i - \bar P_{i+1})(P_{i-1} - P_{i})}{(P_i - P_{i+1})(\bar P_{i-1} - \bar P_{i})} = \frac{(\bar P_i - \lambda_{i,i+1}\bar P_{i+1})(P_{i-1} - \lambda_{i-1,i} P_{i})}{(P_i - \lambda_{i,i+1} P_{i+1})(\bar P_{i-1} - \lambda_{i-1,i}\bar P_{i})}. \label{eq:lincondi}
	\end{align} 
	The product of the three instances of this equation yields the trivial equation $1=1$, which is a proof of Miquel's theorem. The claim in the theorem is equivalent to
	\begin{align}
		\cro(P_{1,2,3},P_1,P_2,P_3) = \overline{\cro(P,P_{2,3},P_{3,1}, P_{1,2})}. \label{eq:crcondi}
	\end{align}
	We may solve Equation~\eqref{eq:lincondi} for $\lambda_{2,3}$ and $\lambda_{3,1}$ and insert these into Equation~\eqref{eq:crcondi}, and also insert the expressions \eqref{eq:givencr}. In this way we obtain a rational equation in the formal variables
	\begin{align}
		P_1,P_2,P_3,\bar P_1,\bar P_2,\bar P_3, \lambda_{1,2}
	\end{align}
	A computer algebra software will then confirm that Equation~\eqref{eq:crcondi} holds.\qed
}

\begin{remark}
	Although Theorem~\ref{th:wconjugate} appears to be rather elementary, we did not find a simpler proof. We were also not able to find the theorem in the literature, although we are not sure that it is new. Note that in terms of Möbius geometry, the theorem states that there is an orientation-reversing Möbius transformation that maps $P,P_{1,2},P_{2,3},P_{3,1}$ to $P_{1,2,3},P_3,P_{1},P_2$.
\end{remark}

Due to Theorem~\ref{th:wtoda}, if the initial $\Wb$-variables are real, then all $\Wb$-variables are real. Moreover, due to Theorem~\ref{th:wconjugate}, if the $\Wb$-variables are real, so are the $\Ww$-variables. Therefore, let us briefly discuss the case of real $W$-variables.

Consider an edge $e=(v,v') = (f,f')^* \in E(\Z^2)$ with $v,f,v',v'$ in counterclockwise order, as well as a circle pattern $(c,p)$ with conical net $t$. We call the \emph{edge cross-ratio} $\gamma(e)$ the cross-ratio
\begin{align}
	\gamma(e) = \cro(t(v), p(f), t(v'), p(f')).
\end{align}
Note that the edge cross-ratio has modulus 1, and the argument is the intersection angle of $c(v)$ and $c(v')$.

\begin{definition}
	An \emph{integrable circle pattern} \cite{bmsanalytic} is a circle pattern such that for each vertex $v\in V$ holds
	\begin{align}
		\gamma(e_1)\gamma(e_2)\gamma(e_3)\gamma(e_4) = 1,
	\end{align}
	where $e_1,e_2,e_3,e_4$ are the four edges containing $v$.
\end{definition}

\begin{theorem}
	Let $(c,p)$ be a Miquel map. The following are equivalent:
	\begin{enumerate}
		\item $\Wb,\Ww$ are real-valued, \label{itm:intcpreal}
		\item $\Wb = \Ww$, \label{itm:intcpequal}
		\item $(c_k,p_k)$ is an integrable circle pattern for all $k\in \Z$.\label{itm:intcp}\qedhere
	\end{enumerate}
\end{theorem}
\proof{
	The equivalence of \ref{itm:intcpreal}.~and \ref{itm:intcpequal}.~is clear due to Theorem~\ref{th:wconjugate}. Next, we explain the equivalence of \ref{itm:intcpreal} and \ref{itm:intcp}. In \cite{amiquel}, it was proven that a Miquel configuration is part of an integrable circle pattern if and only if the four intersection points $\ub(z), \shi{\sm 2} \ub(\shi{\sm 1,\sm 2,\sm 3} z),  \ub(\shi{\sm 3} z),  \ub(\shi{\sm 1} z)$ are on a circle, or equivalently the four intersection points $ \uw(\shi{\sm 1,\sm 2} z), \uw(\shi{\sm 2,\sm 3} z), \uw(z), \uw(\shi{\sm 1,\sm 3} z)$ are on a circle. The $\Wb$-variables are minus the cross-ratio of the first four points, which is real if and only if these points are on a circle, which proves the claim.\qed
}

Note that in \cite{klrr} the fixed points of Miquel dynamics were also considered. In particular, they showed that orthogonal circle patterns and certain circle patterns with constant intersection angles are fixed points of Miquel dynamics. It is not hard to see that these are special cases of integrable circle patterns. Another special case of integrable circle patterns are \emph{isoradial graphs}, see \cite{btisosurvey}. 

\begin{remark}\label{rem:miqintcpzigzag}
	There is an equivalent definition of integrable circle patterns, requiring that the edge cross-ratios $\gamma$ factorize onto zig-zag paths of $\Z^2$. More specifically, a zig-zag is a path that turns maximally right and then maximally left in alternating fashion. To each zig-zag we associate a complex number in $\mathbb S^1 \subset \C$. In an integrable circle pattern, each cross-ratio $\gamma(e)$ is the ratio of the two complex numbers associated to the zig-zags. The effect of Miquel dynamics on an integrable circle pattern is to interchange the complex numbers associated to zig-zags in every second adjacent pair.
\end{remark}

\begin{remark}
	Note that the space of integrable circle patterns is significantly larger than the space of isoradial graphs. Both come with complex numbers associated to the zig-zags as mentioned in Remark~\ref{rem:miqintcpzigzag}. With a fixed choice of these complex numbers, an isoradial graph is completely determined up to translation, rotation and scaling of the plane. On the other hand, with a fixed choice of the complex numbers, an integrable circle pattern is not yet determined even up to translation, rotation and scaling of the plane. Instead, one may still arbitrarily choose one-dimensional boundary data, for example on the coordinate axes.
\end{remark}

\begin{remark}
	One can also consider circle patterns $(c,p)$ on $\Z^3$, for which we switch the role of faces and vertices. That is we consider maps 
	\begin{align}
		p: \Z^3 \rightarrow \hC,\quad  c: F(\Z^3) \rightarrow \mbox{Circles}(\hC),
	\end{align}
	such that $p(v) \in c(f)$ whenever $v\in f$. By using Miquel's theorem it is possible to show that these circle patterns exist and are uniquely defined from 2-dimensional boundary data given on the coordinate planes of $\Z^3$. Theorem~\ref{th:wconjugate} tells us that propagation from boundary data is governed by
	\begin{align}
		\cro(p(z),  p(\shi{1,3} z),  p(\shi{1,2} z),  p(\shi{2,3} z)) =  \overline{\cro(p(\shi{1,2,3}z),  p(\shi{2} z),  p(\shi{3} z),  p(\shi{1} z))}.
	\end{align}
	The same equation without the complex conjugation on the right-hand-side is known as the \emph{double cross-ratio equation} \cite{nsdoublecr}. As a consequence, if we take a circle pattern on $\Z^3$ and apply complex conjugation to every second point, the result will be a solution of the double cross-ratio equation.
\end{remark}

\section{Harmonic maps, orthodiagonal quads, $h$-Miquel maps} \label{sec:harmonic}

\begin{definition} \label{def:hmiquelmap}
	An \emph{$h$-Miquel map} is a Miquel map $(c,p)$ such that for each $z \in \Z^3_-$ with $z_3 = 0$ the four circle centers $t(\shi3 z)$, $t(\shi{1} z)$, $t(\shi{2} z)$, $t(\shi{1,2,3} z)$ form a rectangle.
\end{definition}

Let $(c_0, p_0)$ be the circle pattern defined by an $h$-Miquel map at level 0, as defined in Equation~\eqref{eq:miquellevel}. Recall that due to Definition~\ref{def:cp} the circles $c_0$ are defined on $\Z^2$ and the points $p_0$ on $F(\Z^2)$. Definition~\ref{def:hmiquelmap} is equivalent to the requirement that around each second face the four circle centers form a rectangle, see also Figure~\ref{fig:orthodiag}. Recall that adjacent intersection points of a circle pattern are related by a reflection about the corresponding adjacent circle centers. As a result, for every $z \in \Z^3_-$ with $z_3=0$, the four points $\ub(\shi{\sm1}z)$, $\ub(\shi{\sm2}z)$, $\ub(\shi{1}z)$, $\ub(\shi{2}z)$ form an orthodiagonal quad, that is a quad with orthogonal diagonals. Moreover, $\uw(z)$ is the intersection point of the diagonals of that orthodiagonal quad. See \cite{josefssonorthodiagonal} for more details on orthodiagonal quads, they also appear in the theory of discrete linear complex analysis \cite{bglinear, bmsanalytic} and the study of random maps \cite{ggjnortho}. Conversely, every map consisting of orthodiagonal quads defines a circle pattern and an $h$-Miquel map.

\begin{figure}
	\centering
	\begin{tikzpicture}[scale=1.8,font=\sffamily]
		\node[bvert, label=right:$\ub(\shi{1} z)$] (o1) at (2,0) {};
		\node[bvert, label=left:$\ub(\shi{\sm1} z)$] (om1) at (-1,0) {};
		\node[bvert, label=right:$\ub(\shi{2} z)$] (o2) at (0,1.5) {};
		\node[bvert, label=right:$\ub(\shi{\sm2} z)$] (om2) at (0,-1.2) {};
		
		\draw[-] 
			(o1) -- (om1)
			(o2) -- (om2)
		;
		
		\node[wvert] (w) at (0,0) {};
		
		\node[bvert, blue, label=left:$t( \shi{3}z)$] (t) at ($0.5*(om1)+0.5*(om2)$) {};
		\node[bvert, blue, label=right:$t(\shi{1} z)$] (t1) at ($0.5*(o1)+0.5*(om2)$) {};
		\node[bvert, blue, label=left:$t(\shi{2} z)$] (t2) at ($0.5*(om1)+0.5*(o2)$) {};
		\node[bvert, blue, label=right:$t(\shi{1,2,3} z)$] (t12) at ($0.5*(o1)+0.5*(o2)$) {};

 		\node [draw, blue] at (t1) [circle through={(w)}] {};
 		\node [draw, blue] at (t2) [circle through={(w)}] {};
 		\node [draw, blue] at (t) [circle through={(w)}] {};
 		\node [draw, blue] at (t12) [circle through={(w)}] {};
 		
 		\coordinate[] (p1) at ($(t1)!(w)!(t2)$) ;
		\node[bvert, label=right:$\ub(\shi{\sm3} z)$] (b3) at ($(w)!2!(p1)$) {};
 		\coordinate[] (p2) at ($(t)!(w)!(t12)$) ;
		\node[bvert, label=right:$\ub(\shi{3} z)$] (bm3) at ($(w)!2!(p2)$) {};
		
		\draw[-, blue]
			(t) -- (t1) -- (t12) -- (t2) -- (t)
		;

		\node[wvert] (w) at (0,0) {};
	\end{tikzpicture}
	\caption{An orthodiagonal quad in an $h$-Miquel map with labels assuming $z\in \Z^3_-$ with $z_3=0$. The unlabeled white point in the center is $\uw(z)$.}
	\label{fig:orthodiag}
\end{figure}

\begin{remark}
	If we restrict $\ub$ to $\{ z \in \Z^3 \ | \  z_3 = 0, z_1 + z_2 \in 4\Z \}$, which is half of the points of $\ub_0$, then this map is called a \emph{harmonic map}. This is because this restriction satisfies a Laplace equation with real coefficients given by quotients of diagonal lengths. Moreover, if $t_0$ is an embedding, then these quotients are real positive, and the restriction of $\ub$ is called a \emph{harmonic embedding} or \emph{Tutte embedding} \cite{tutteembedding}. The restriction of $\ub$ to  $\{ z \in \Z^3 \ | \  z_3 = 0, z_1 + z_2 \in 4\Z + 2 \}$ is the dual harmonic map that satisfies a Laplace equation with inverse coefficients. See also \cite{clrtembeddings, klrr} for more details.
\end{remark}

\begin{lemma}\label{lem:focalformula}
	Let $(c,p)$ be an $h$-Miquel map. For every $z \in \Z^3_-$ with $z_3 = 0$ holds 
	\begin{align}
		\ub(\shi{3}z) &= \frac{\ub(\shi{1}z) \ub(\shi{2}z) - \ub(\shi{\sm1}z) \ub(\shi{\sm2}z)}{\ub(\shi{1}z) + \ub(\shi{2}z) - \ub(\shi{\sm1}z) - \ub(\shi{\sm2}z)}, \label{eq:focalformone} \\
		\ub(\shi{\sm 3}z) &= \frac{\ub(\shi{1}z) \ub(\shi{\sm 2}z) - \ub(\shi{\sm1}z) \ub(\shi{2}z)}{\ub(\shi{1}z) + \ub(\shi{\sm2}z) - \ub(\shi{\sm1}z) - \ub(\shi{2}z)}.  \label{eq:focalformtwo}
	\end{align}		
\end{lemma}

\proof{
	A quick calculation shows that Equation~\eqref{eq:focalformone} is equivalent to
	\begin{align}
		\mr(\ub(\shi{3}z), \ub(\shi{1}z), \ub(\shi{\sm1}z) , \infty, \ub(\shi{2}z), \ub(\shi{\sm2}z)) = -1.
	\end{align}
	Additionally, we note that the six points involved in this multi-ratio form a Clifford configuration, see also the proof of Theorem~\ref{th:udskp} and Figure~\ref{fig:cpathree}. Therefore, as explained in the proof of Theorem~\ref{th:udskp}, the multi-ratio equation holds. The proof of Equation~\eqref{eq:focalformtwo} proceeds analogously.\qed
}

\begin{theorem} \label{th:harmxy}
	The $\Xb$-variables and $Y$-variables coincide, that is $Y(z) = \Xb(z)$ for all $z\in \lat$.
\end{theorem}
\proof{
	Using Lemma~\ref{lem:focalformula} and a quick calculation shows that
	\begin{align}
		\frac{t(\shi{123}z) - t(\shi{2}z)}{t(\shi{3}z) - t(\shi{2}z)} &= - \frac{\ub(\shi{\sm 1}z) - \ub(\shi{\sm 3}z)}{\ub(\shi{\sm 3}z) - \ub(\shi{2}z)}, &
		\frac{t(\shi{3}z) - t(\shi{1}z)}{t(\shi{123}z) - t(\shi{1}z)} &= - \frac{\ub(\shi{\sm 2}z) - \ub(\shi{\sm 3}z)}{\ub(\shi{\sm 3}z) - \ub(\shi{1}z)}, \\
		\frac{t(\shi{2}z) - t(\shi{123}z)}{t(\shi{1}z) - t(\shi{123}z)} &= - \frac{\ub(\shi{1}z) - \ub(\shi{ 3}z)}{\ub(\shi{ 3}z) - \ub(\shi{2}z)}, &
		\frac{t(\shi{1}z) - t(\shi{3}z)}{t(\shi{2}z) - t(\shi{3}z)} &= - \frac{\ub(\shi{\sm 1}z) - \ub(\shi{3}z)}{\ub(\shi{3}z) - \ub(\shi{\sm 2}z)}.
	\end{align}
	Each $Y$-variable is defined by two ratios of the kind on the left of the equalities, while each corresponding $\Xb$-variable is defined by two ratios of the kind on the right of the equalities, which proves the claim for initial data. Since $\Xb$-variables and $Y$-variables both constitute a Y-system, they have the same recurrence relation and therefore the claim holds in general.
	\qed
}

Theorem~\ref{th:harmxy} is in a sense surprising, since $\Xb$- and $Y$-variables have different transformation behaviour. In particular, applying a Möbius transformation to an $h$-Miquel map does not yield an $h$-Miquel map. Moreover, the theorem states that even though only the 0-th layer in an $h$-Miquel map needs to consist of rectangles, the equality of $\Xb$- and $Y$-variables holds for the whole $h$-Miquel map, which leads us to the following question.

\begin{question}
	Is there a local geometric characterization of circle patterns such that $\Xb$- and $Y$-variables coincide?
\end{question}

Note that, as explained in Section~\ref{sec:toda}, the $\Xb$- and the $Y$-variables determine the circle pattern up to boundary data. Hence, there is definitely a geometric constraint, the question is whether it may be recognized locally (away from the boundary data), and if so how. 

Let us also add that the $Y$-variables -- considered as cluster variables -- were shown in \cite{klrr} to be in the so called \emph{resistor subvariety}, which was defined in \cite{gkdimers} and is based on Temperley's bijection \cite{temperleytrick, kpwtrees} which represents spanning tree statistics as a special case of dimer statistics. Since $Y$- and $\Xb$-variables coincide, we get the following corollary for free

\begin{corollary}
	The 0-th layer of the $\Xb$ variables is in the resistor subvariety. More explicitly
	\begin{align}
		\Xb(\shi{3}(z)) \Xb(\shi{123}(z))  = \Xb(\shi{1}(z))  \Xb(\shi{2}(z)),
	\end{align}		
	for all $z\in \lat$ with $z_3 = 0$.
\end{corollary}

\section{Circle packings, $s$-embeddings, $s$-Miquel maps} \label{sec:packings}

\begin{definition} \label{def:smiquelmap}
	An \emph{$s$-Miquel map} is a Miquel map $(c,p)$ such that $c(z)$ is tangent to $c(z')$ whenever $|z-z'| = \sqrt{2}$ and $z_3 = z'_3 = 0$.
\end{definition}

\begin{figure}
	\centering
	\includegraphics[scale=0.65]{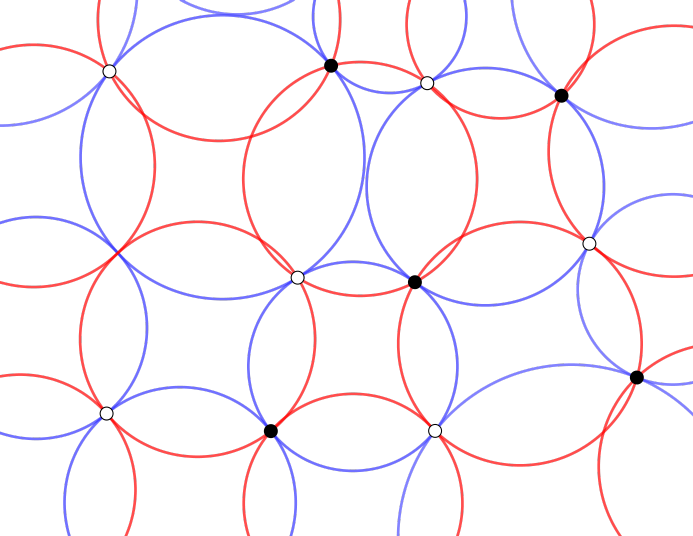}
	\hspace{2mm}
	\includegraphics[scale=0.65]{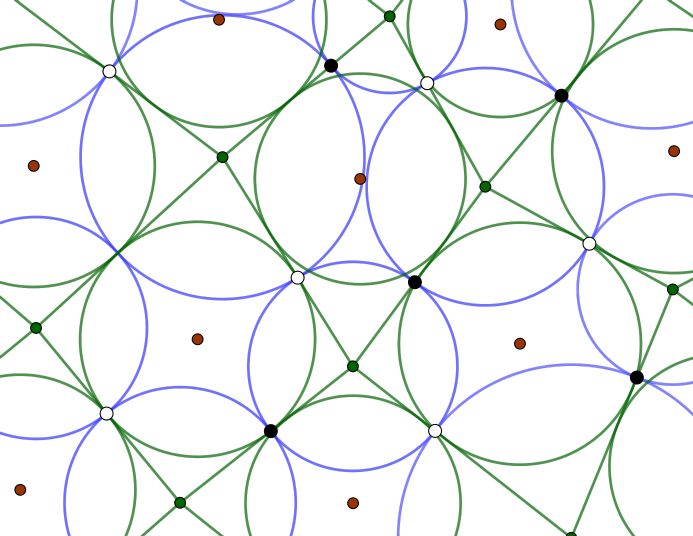}
	\caption{Left: the blue circles are a circle packing. Right: the blue circles are the same as on the left The green points are the centers of the blue circles and are called the $s$-embedding. The green quads have green incircles, the green circles are concentric with the red circles.}
	\label{fig:circlepacking}
\end{figure}

The restriction of an $s$-Miquel map to the circles $c(z)$ with $z_3 = 0$ is called a \emph{circle packing}. The $t$-map of $(c_0,p_0)$ is called an \emph{$s$-embedding} (with some additional embedding requirements, see \cite{chelkaksembeddings}).

\begin{remark}
	Note that due to the so called \emph{touching coins lemma}, every circle packing gives rise to an $s$-Miquel map.
\end{remark}

The definition of $s$-Miquel maps implies that $p(i,j,0) = p(i,j,-1)$ for all $i,j \in \Z$, since these are the tangency points. As a result $(c_0,p_0)$ coincides with $(c_{-1},p_{-1})$. Therefore forwards Miquel dynamics starting from $(c_0,p_0)$ generates the same circle patterns as backwards Miquel dynamics starting from $(c_{-1},p_{-1})$ does. Hence there is a global symmetry in $s$-Miquel maps, namely that
\begin{align}
	 p((\shi{3})^k z) &=  p((\shi{\sm 3})^{k+1} z), & c( (\shi{3})^k z) &=  c((\shi{\sm 3})^{k} z),
\end{align}
for all $z$ in $\Z^3$ resp.~$\lat$ with $z_3=0$. In terms of $Y$-variables, we obtain the following theorem.

\begin{theorem}\label{th:ssymmetryy}
	In an $s$-Miquel map the $Y$-variables satisfy 	
	\begin{align}
		Y((\shi{3})^k z) =  Y^{-1}((\shi{\sm3})^{k} z),
	\end{align}
	for all $z \in \lat$ with $z_3 = -1$ and $k\in \Z$, $k > 0$. 
\end{theorem}
\proof{
	Combining the symmetries in the $s$-Miquel map with Definition~\ref{def:yvariables} and Equation~\eqref{eq:yy} yields the claim. \qed
}

Setting $k=1$, Theorem~\ref{th:ssymmetryy} combined with the propagation equation for $Y$-variables (Equation~\eqref{eq:ysystems}) puts a constraint on the $Y$-variables $Y(z)$ with $z_3 = 0$. This constraint is known as the \emph{Ising subvariety}, see \cite{kenyonpemantle}.

In terms of $X$- and $W$-variables, we obtain the following theorem.

\begin{theorem}\label{th:ssymmetryxw}
	In an $s$-Miquel map the $X$ and $W$-variables satisfy 	
	\begin{align}
		\Xb((\shi{3})^k z) &=  \Xw((\shi{\sm3})^{k} z),\\
		\Wb((\shi{3})^k z) &=  \Ww((\shi{\sm3})^{k} z),
	\end{align}
	for all $z \in \lat$ with $z_3 = -1$ and $k\in \Z$. 
\end{theorem}
\proof{
	Combining the symmetries in the Miquel map with Definition~\ref{def:xvariables} or Definition~\ref{def:wvariables} respectively yields the claim. \qed
}

Unlike in the case of the $Y$-variables, it is not clear whether there is an analogue to the Ising subvariety for the $X$- and $W$-variables.

\begin{remark}
	Let us add a final remark for the reader familiar with cube-flips in cluster algebras or the Ising model.
	In the case of $s$-embeddings on general quad-graphs there is a cube-flip, which corresponds to a sequence of mutations in the cluster structure, see \cite{kenyonpemantle, mrtcube}. Since the $Y$-variables are in the Ising subvariety, the recurrence equation for the $Y$-variables is called the \emph{Kashaev recurrence}, a special case of the \emph{hexahedron recurrence} \cite{kenyonpemantle}. We claim here without proof that the corresponding $X$-variables satisfy the hexahedron recurrence as well.
\end{remark}

\bibliographystyle{plain}
\bibliography{references}

\end{document}